\documentclass{vldb}
\usepackage{color}
\usepackage{multirow}
\usepackage{url}
\usepackage{amsmath, bbold,xspace}
\usepackage{multibib}
\usepackage{balance}

\usepackage[T1]{fontenc}

\usepackage{algorithmicx,algorithm}
\usepackage{algpseudocode}

\floatname{algorithm}{Algorithm}

\newcommand{\algrule}[1][.2pt]{\par\vskip.5\baselineskip\hrule height #1\par\vskip.5\baselineskip}

\makeatletter
 \def\@textbottom{\vskip \z@ \@plus 1pt}
 \let\@texttop\relax
\makeatother

\newcommand{\eat}[1]{}
\newcommand{\systemx}{\textrm{DeepDive}\xspace}
\newcommand{\D}{\mathbb{D}}
\DeclareMathOperator\ts{:-}

\usepackage{eulervm,cite}

\newcites{app,cover}{References,References}

\newtheorem{theorem}{Theorem}[section]

\newtheorem{proposition}[theorem]{Proposition}

\newtheorem{example}[theorem]{Example}
\newtheorem{definition}[theorem]{Definition}

\def\compactify{\itemsep=0pt \topsep=0pt \partopsep=0pt \parsep=0pt}

\title{Incremental Knowledge Base Construction Using DeepDive}

\author{
Jaeho Shin$^\dagger$~~~~Sen Wu$^\dagger$~~~~Feiran Wang$^\dagger$~~~~Christopher De Sa$^\dagger$~~~~Ce Zhang$^\dagger$$^\ddagger$~~~~Christopher R\'{e}$^\dagger$\\
$^\dagger$Stanford University\\
$^\ddagger$University of Wisconsin-Madison\\
\{jaeho, senwu, feiran, cdesa, czhang, chrismre\}@cs.stanford.edu}

\begin{document}
\maketitle

\begin{abstract}
Populating a database with unstructured information is a long-standing
problem in industry and research that encompasses problems of
extraction, cleaning, and integration. Recent names used for this
problem include dealing with dark data and knowledge base construction
(KBC). In this work, we describe \systemx, a system that combines
database and machine learning ideas to help develop KBC systems, and
we present techniques to make the KBC process more efficient. We
observe that the KBC process is iterative, and we develop techniques
to incrementally produce inference results for KBC systems. We propose
two methods for incremental inference, based respectively on sampling
and variational techniques. We also study the tradeoff space of these
methods and develop a simple rule-based optimizer. \systemx includes
all of these contributions, and we evaluate \systemx on five KBC
systems, showing that it can speed up KBC inference tasks by up to two
orders of magnitude with negligible impact on quality.
\end{abstract}

\section{Introduction}\label{sec:intro}
The process of populating a structured relational database from
unstructured sources has received renewed interest in the database
community through high-profile start-up companies (e.g., Tamr and
Trifacta), established companies like IBM's
Watson~\cite{Ferrucci:2010:AI,Brown:2013:IBM}, and a variety of
research
efforts~\cite{Chen:2014:SIGMOD,Nakashole:2011:WSDM,Weikum:2010:PODS,Li:2011:ACL,Shen:2007:VLDB}. At
the same time, communities such as natural language processing and
machine learning are attacking similar problems under the name {\em
  knowledge base construction}
(KBC)~\cite{Jiang:2012:ICDM,Dong:2014:VLDB,Betteridge:2009:AAAI}. While
different communities place differing emphasis on the extraction,
cleaning, and integration phases, all communities seem to be converging
toward a common set of techniques that include a mix of data
processing, machine learning, and engineers-in-the-loop.

The ultimate goal of KBC is to obtain high-quality structured data
from unstructured information. These databases are richly structured
with tens of different entity types in complex
relationships. Typically, quality is assessed using two complementary
measures: precision (how often a claimed tuple is correct) and recall
(of the possible tuples to extract, how many are actually
extracted). These systems can ingest massive numbers of documents--far
outstripping the document counts of even well-funded human
curation efforts. Industrially, KBC systems are constructed by skilled
engineers in a months-long (or longer) process--not a one-shot
algorithmic task. Arguably, the most important question in such
systems is how to best use skilled engineers' time to rapidly improve
data quality. In its full generality, this question spans a number of
areas in computer science, including programming languages, systems,
and HCI. We focus on a narrower question, with the axiom that {\it the
  more rapid the programmer moves through the KBC construction loop,
  the more quickly she obtains high-quality data.}

This paper presents DeepDive, our open-source engine for knowledge
base construction.\footnote{\url{http://deepdive.stanford.edu}}
DeepDive's language and execution model are similar to other KBC
systems: DeepDive uses a high-level declarative
language~\cite{Niu:2011:VLDB,Chen:2014:SIGMOD,Nakashole:2011:WSDM}.
From a database perspective, DeepDive's language is based on SQL. From
a machine learning perspective, DeepDive's language is based on Markov
Logic~\cite{Niu:2011:VLDB,Domingos:2009:Book}: DeepDive's language
inherits Markov Logic Networks' (MLN's) formal
semantics.\footnote{DeepDive has some technical differences from
  Markov Logic that we have found useful in building applications. We
  discuss these differences in Section~\ref{sec:design}.}  Moreover,
it uses a standard execution model for such
systems~\cite{Niu:2011:VLDB,Chen:2014:SIGMOD,Nakashole:2011:WSDM} in
which programs go through two main phases: {\em grounding}, in which
one evaluates a sequence of SQL queries to produce a data structure
called a {\em factor graph} that describes a set of random variables
and how they are correlated. Essentially, every tuple in the database
or result of a query is a random variable (node) in this factor
graph. The {\em inference} phase takes the factor graph from grounding
and performs statistical inference using standard techniques, e.g.,
Gibbs sampling~\cite{Zhang:2013:SIGMOD,Wick:2010:PVLDB}. The output of
inference is the marginal probability of every tuple in the
database. As with Google's Knowledge Vault~\cite{Dong:2014:VLDB} and
others~\cite{Niu:2012:IJSWIS}, DeepDive also produces marginal
probabilities that are {\em calibrated}: if one examined all facts
with probability $0.9$, we would expect that approximately $90\%$ of
these facts would be correct. To calibrate these probabilities,
DeepDive estimates (i.e., learns) parameters of the statistical model
from data. Inference is a subroutine of the learning procedure and is
the critical loop. Inference and learning are computationally intense
(hours on 1TB RAM/48-core machines).

In our experience with DeepDive, we found that KBC is an iterative
process. In the past few years, DeepDive has been used to build dozens
of high-quality KBC systems by a handful of technology companies, a
number law enforcement agencies via DARPA's MEMEX, and scientists in
fields such as paleobiology, drug repurposing, and genomics. Recently,
we compared a DeepDive system's extractions to the quality of
extractions provided by human volunteers over the last ten years for a
paleobiology database, and we found that the DeepDive system had
higher quality (both precision and recall) on many entities and
relationships. Moreover, on all of the extracted entities and
relationships, DeepDive had no worse
quality~\cite{Peters:2014:ArXiv}. Additionally, the winning entry of
the 2014 TAC-KBC competition was built on
DeepDive~\cite{Angeli:2014:TAC}. In all cases, we have seen the
process of developing KBC systems is iterative: quality requirements
change, new data sources arrive, and new concepts are needed in the
application. This led us to develop techniques to make the entire
pipeline incremental in the face of changes both to the data and to
the DeepDive program. Our primary technical contributions are to make
the grounding and inference phases more incremental.\footnote{ As
  incremental learning uses standard techniques, we cover it only in
  the full version of this paper. }

\paragraph*{Incremental Grounding}
Grounding and feature extraction are performed by a series of SQL
queries. To make this phase incremental, we adapt the algorithm of
Gupta, Mumick, and Subrahmanian~\cite{Gupta:1993:SIGMOD}.  In
particular, DeepDive allows one to specify ``delta rules'' that
describe how the output will change as a result of changes to the
input. Although straightforward, this optimization has not been
applied systematically in such systems and can yield up to 360$\times$
speedup in KBC systems.

\paragraph*{Incremental Inference} Due to our choice of incremental
grounding, the input to DeepDive's inference phase is a factor graph
along with a set of changed data and rules. The goal is to compute the
output probabilities computed by the system. Our approach is to frame
the incremental maintenance problem as one of approximate
inference. Previous work in the database community has looked at how
machine learning data products change in response to both to new
labels~\cite{Koc:2011:PVLDB} and to new
data~\cite{Chen:2008:ICDE,Chen:2012:ICDE}. In KBC, both the program
and data change on each iteration. Our proposed approach can cope with
both types of change simultaneously.

The technical question is which approximate inference algorithms to
use in KBC applications. We choose to study two popular classes of
approximate inference techniques: \emph{sampling-based
  materialization} (inspired by sampling-based probabilistic databases
such as MCDB~\cite{Jampani:2008:SIGMOD}) and \emph{variational-based
  materialization} (inspired by techniques for approximating graphical
models~\cite{Wainwright:2006:TSP}). Applying these techniques to
incremental maintenance for KBC is novel, and it is not theoretically
clear how the techniques compare. Thus, we conducted an experimental
evaluation of these two approaches on a diverse set of DeepDive
programs.

We found these two approaches are sensitive to changes along
three largely orthogonal axes: the size of the factor graph, the
sparsity of correlations, and the anticipated number of future
changes. The performance varies by up to two orders of magnitude in
different points of the space. Our study of the tradeoff space
highlights that neither materialization strategy dominates the
other. To automatically choose the materialization strategy, we
develop a simple rule-based optimizer.

\paragraph*{Experimental Evaluation Highlights} We used DeepDive programs
developed by our group and \systemx users to understand whether the
improvements we describe can speed up the iterative development
process of DeepDive programs. To understand the extent to which
\systemx's techniques improve development time, we took a sequence of
six snapshots of a KBC system and ran them with our incremental
techniques and completely from scratch. In these snapshots, our
incremental techniques are 22$\times$ faster. The results for each
snapshot differ at most by 1\% for high-quality facts (90\%+
accuracy); fewer than 4\% of facts differ by more than $0.05$ in
probability between approaches. Thus, essentially the same facts were
given to the developer throughout execution using the two techniques,
but the incremental techniques delivered them more quickly.

\paragraph*{Outline} The rest of the paper is organized as follows.
Section~\ref{sec:lang_main} contains an in-depth analysis of the KBC
development process, and the presentation of our language for modeling
KBC systems. We discuss the different techniques for incremental
maintenance in Section~\ref{sec:tradeoff}. We also present the results
of the exploration of the tradeoff space and the description of our
optimizer. Our experimental evaluation is presented in
Section~\ref{sec:exp}.

\subsection*{Related Work}\label{sec:related_work}

\textbf{Knowledge Base Construction (KBC)} KBC has been an area of
intense study over the last decade, moving from pattern
matching~\cite{Hearst:1992:COLING} and rule-based
systems~\cite{Li:2011:ACL} to systems
that use machine learning for
KBC~\cite{Etzioni:2004:WWW,Betteridge:2009:AAAI,
Carlson:2010:AAAI,Nakashole:2011:WSDM,Dong:2014:VLDB}. Many groups have
studied how to improve the quality of specific components of KBC
systems~\cite{mintz2009distant,yao2010collective}.
We build on this line of work. We formalized the development process
and built \systemx to ease and accelerate the KBC process, which we
hope is of interest to many of these systems as well. DeepDive has
many common features to Chen and Wang~\cite{Chen:2014:SIGMOD},
Google's Knowledge Vault~\cite{Dong:2014:VLDB}, and a forerunner of
DeepDive, Tuffy~\cite{Niu:2011:VLDB}. We focus on the incremental
evaluation from feature extraction to inference.\\[-2pt]

\noindent
\textbf{Declarative Information Extraction} The database community has
proposed declarative languages for information extraction, a task with
similar goals to knowledge base construction, by extending relational
operations~\cite{Gottlob:2004:PODS,Shen:2007:VLDB,Li:2011:ACL},
or rule-based approaches~\cite{Nakashole:2011:WSDM}. These approaches
can take advantage of classic view maintenance techniques to make the
execution incremental, but they do not study how to incrementally
maintain the result of statistical inference and learning, which is
the focus of our work.\\[-2pt]

\noindent
\textbf{Incremental Maintenance of Statistical Inference and Learning}
Related work has focused on incremental inference for specific classes
of graphs (tree-structured~\cite{Delcher:1996:JAIR,Acar:NIPS:2007} or
low-degree~\cite{Acar:2008:UAI} graphical models). We deal instead
with the class of factor graphs that arise from the KBC process, which
is much more general than the ones examined in previous approaches.
Nath and Domingos~\cite{Nath:2010:AAAI} studied how to extend belief
propagation on factor graphs with new evidence, but without any
modification to the structure of the graph. Wick and
McCallum~\cite{Wick:2011:NIPS} proposed a ``query-aware MCMC''
method. They designed a proposal scheme so that query variables tend to
be sampled more frequently than other variables. We frame our problem
as approximate inference, which allows us to handle changes to the
program and the data in a single approach.\\[-2pt]

\section{KBC using DeepDive}\label{sec:lang_main}

We describe \systemx, an end-to-end framework for building KBC systems
with a declarative language. We first recall standard definitions, and
then introduce the essentials of the framework by example, compare our
framework with Markov Logic, and describe DeepDive's formal semantics.

\begin{figure}[t]
\centering
\includegraphics[width=\columnwidth]{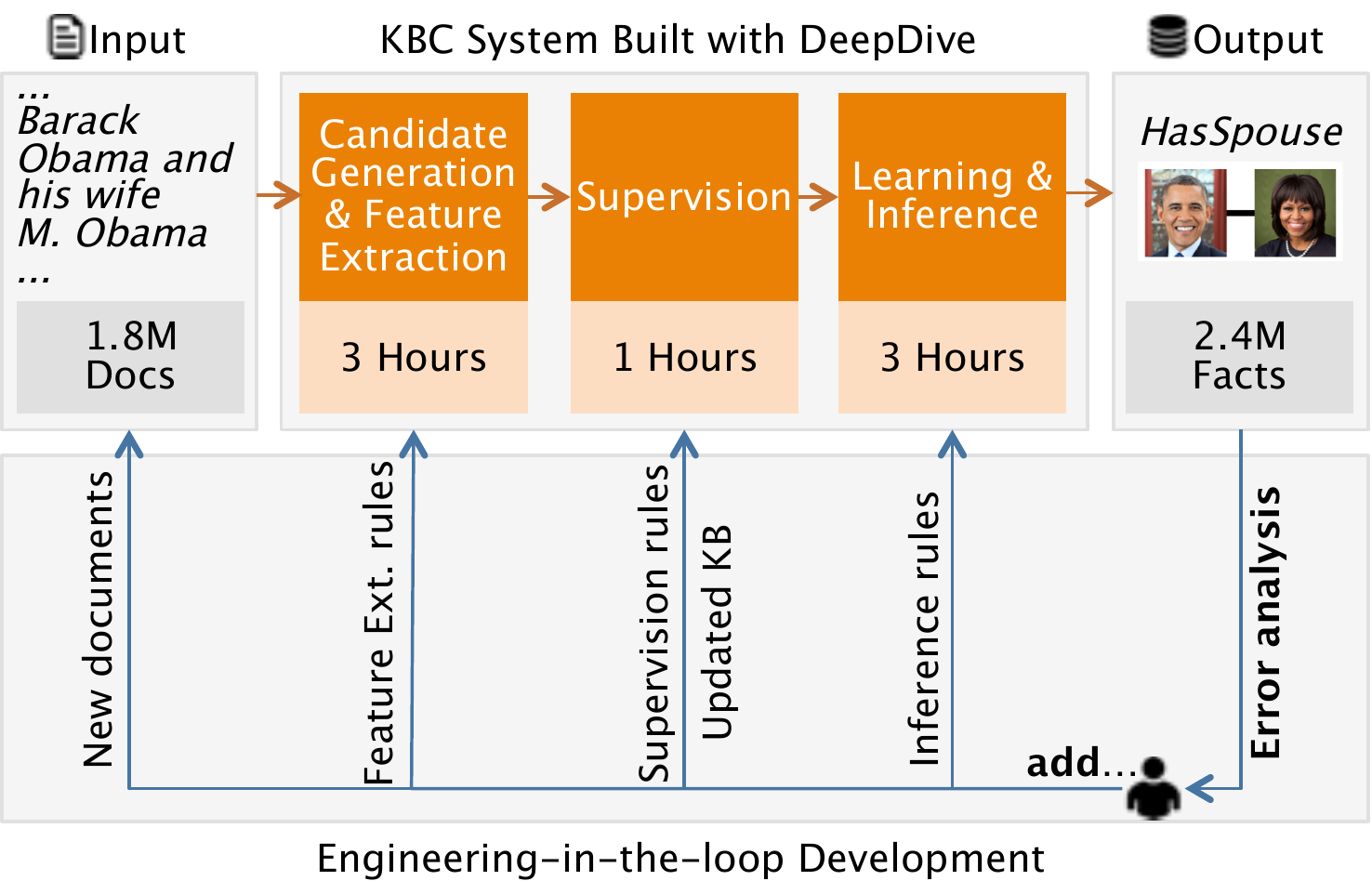}
\caption{A KBC system takes as input unstructured documents and
  outputs a structured knowledge base. The runtimes are for the
  TAC-KBP competition system (News). To improve quality, the developer
  adds new rules and new data.}
\label{fig:loop}
\end{figure}

\subsection{Definitions for KBC Systems}

The \emph{input} to a KBC system is a heterogeneous collection of
unstructured, semi-structured, and structured data, ranging from text
documents to existing but incomplete KBs. The \emph{output} of the
system is a relational database containing facts extracted from the
input and put into the appropriate schema. Creating the knowledge base
may involve extraction, cleaning, and integration.

\begin{example} Figure~\ref{fig:loop}
illustrates our running example: a knowledge base with pairs of
individuals that are married to each other. The input to the system is
a collection of news articles and an incomplete set of married
persons; the output is a KB containing pairs of person that are
married. A KBC system extracts linguistic patterns, e.g., ``...~and
his wife~...''  between a pair of mentions of individuals (e.g.,
``Barack Obama'' and ``M. Obama''). Roughly, these patterns are then
used as features in a classifier deciding whether this pair of mentions
indicates that they are married (in the \textsf{HasSpouse}) relation.
\label{example:kbc}
\end{example}

We adopt standard terminology from KBC, e.g.,
ACE.\footnote{\url{http://www.itl.nist.gov/iad/mig/tests/ace/2000/}}
There are four types of objects that a KBC system seeks to extract
from input documents, namely \emph{entities}, \emph{relations},
\emph{mentions}, and \emph{relation mentions}.  An {\bf entity} is a
real-world person, place, or thing. For example,
``Michelle\_Obama\_1'' represents the actual entity for a person whose
name is ``Michelle Obama''; another individual with the same name
would have another number. A {\bf relation} associates two (or more)
entities, and represents the fact that there exists a relationship
between the participating entities.  For example, ``Barack\_Obama\_1''
and ``Michelle\_Obama\_1'' participate in the \textsf{HasSpouse}
relation, which indicates that they are married. These real-world
entities and relationships are described in text; a {\bf mention} is a
span of text in an input document that refers to an entity or
relationship: ``Michelle'' may be a mention of the entity
``Michelle\_Obama\_1.'' A \emph{relation mention} is a phrase that
connects two mentions that participate in a relation such as
``(Barack Obama, M. Obama)". The process of mapping mentions to
entities is called \emph{entity linking}.

\subsection{The DeepDive Framework} \label{sec:dd}

\begin{figure*}[t]
\centering
\includegraphics[width=1.0\textwidth]{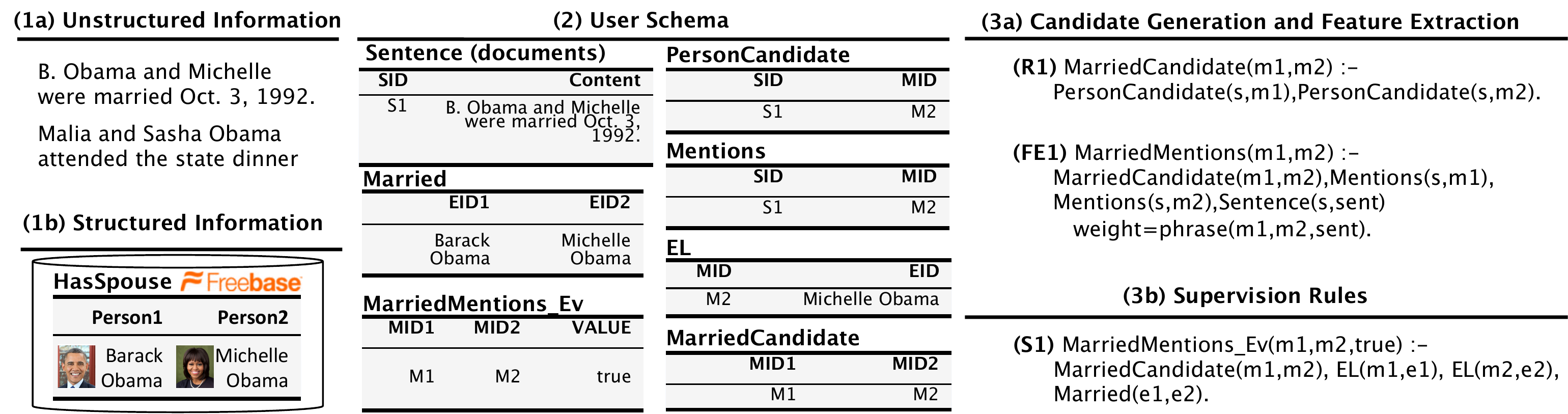}
\caption{An example KBC system. See Section~\ref{sec:dd} for details.}
\label{fig:language}
\end{figure*}

\systemx is an end-to-end framework for building KBC systems, as shown
in Figure~\ref{fig:loop}.\footnote{For more information, including
  examples, please see \url{http://deepdive.stanford.edu}. Note that
  our engine is built on Postgres and Greenplum for all SQL processing
  and UDFs. There is also a port to MySQL.} We walk through each
phase. DeepDive supports both SQL and datalog, but we use datalog syntax
for exposition.
The rules we describe in this section are manually created 
by the user of DeepDive and the process of creating these rules is 
application-specific.

\paragraph*{Candidate Generation and Feature Extraction}

All data in \systemx is stored in a relational database. The first
phase populates the database using a set of SQL queries and
user-defined functions (UDFs) that we call {\em feature
  extractors}. By default, \systemx stores all documents in the
database in one sentence per row with markup produced by standard NLP
pre-processing tools, including HTML stripping, part-of-speech
tagging, and linguistic parsing. After this loading step, \systemx
executes two types of queries: (1) {\it candidate mappings}, which are
SQL queries that produce possible mentions, entities, and relations,
and (2) {\it feature extractors} that associate features to
candidates, e.g., ``...~and his wife~...''  in
Example~\ref{example:kbc}.

\begin{example}
Candidate mappings are usually simple. Here, we create a relation
mention for every pair of candidate persons in the same sentence
($s$):
\begin{align*}
        \mbox{\texttt{(R1) }}&\texttt{MarriedCandidate}(m1,m2)\ts\\
	&\texttt{PersonCandidate}(s,m1),\texttt{PersonCandidate}(s,m2).
\end{align*}
\end{example}

Candidate mappings are simply SQL queries with UDFs that look like
low-precision but high-recall ETL scripts. 
Such rules must be high
recall: if the union of candidate mappings misses a fact, DeepDive has
no chance to extract it. 

We also need to extract features, and we extend classical Markov Logic
in two ways: (1) user-defined functions and (2) weight tying, which
we illustrate by example.

\begin{example}
Suppose that $\mathsf{phrase}(m1,m2,sent)$ returns the phrase between
two mentions in the sentence, e.g., ``and his wife'' in the above
example. The phrase between two mentions may indicate whether two
people are married. We would write this as:
\begin{align*}
  \mbox{\texttt{(FE1) }}&\texttt{MarriedMentions}(m1,m2)\ts\\
&\texttt{MarriedCandidate}(m1,m2),\texttt{Mention}(s,m1),\\
  &\texttt{Mention}(s,m2),
  \texttt{Sentence}(s,sent) \\
  &weight = \mathsf{phrase}(m1,m2,sent).
\end{align*}
One can think about this like a classifier: This rule says that
whether the text indicates that the mentions $m1$ and $m2$ are married
is {\em influenced} by the phrase between those mention pairs. The
system will infer based on training data its confidence (by estimating
the weight) that two mentions are indeed indicated to be married.
\end{example}

Technically, $\mathsf{phrase}$ returns an identifier that determines
which weights should be used for a given relation mention in a
sentence.  If $\mathsf{phrase}$ returns the same result for two
relation mentions, they receive the {\em same} weight. We explain
weight tying in more detail in Section \ref{sec:design}.  In general,
$\mathsf{phrase}$ could be an arbitrary UDF that operates in a
per-tuple fashion. This allows DeepDive to support common examples of
features such as ``bag-of-words'' to context-aware NLP features to
highly domain-specific dictionaries and ontologies. In addition to
specifying sets of classifiers, DeepDive inherits Markov Logic's
ability to specify rich correlations between entities via weighted
rules. Such rules are particularly helpful for data cleaning and data
integration.

\paragraph*{Supervision}
Just as in Markov Logic, \systemx can use training data or evidence
about any relation; in particular, each user relation is associated
with an evidence relation with the same schema and an additional field
that indicates whether the entry is true or false. Continuing our
example, the evidence relation $MarriedMentions\_Ev$ could contain
mention pairs with positive and negative labels. Operationally, two
standard techniques generate training data: (1) hand-labeling, and (2)
{\em distant supervision}, which we illustrate below.

\begin{example}
\emph{Distant
  supervision}~\cite{mintz2009distant,DBLP:conf/acl/HoffmannZLZW11} is
a popular technique to create evidence in KBC systems. The idea is to
use an incomplete KB of married entity pairs to {\em heuristically}
label (as \texttt{True} evidence) all relation mentions that link to a
pair of married entities:
\begin{align*}
\mbox{\texttt{(S1) }}&\texttt{MarriedMentions\_Ev}(m1,m2,true) \ts\\ 
&\texttt{MarriedCandidates}(m1,m2),
  \texttt{EL}(m1,e1),\\ 
&\texttt{EL}(m2,e2), \texttt{Married}(e1,e2).
\end{align*}
Here, \texttt{Married} is an (incomplete) list of married real-world
persons that we wish to extend.  The relation EL is for ``entity
linking'' that maps mentions to their candidate entities. At first
blush, this rule seems incorrect. However, it generates noisy,
imperfect examples of sentences that indicate two people are
married. Machine learning techniques are able to exploit redundancy to
cope with the noise and learn the relevant phrases (e.g., ``and his
wife'' ). Negative examples are generated by relations that are
largely disjoint (e.g., siblings). Similar to
  DIPRE\cite{Brin:1998:WebDB} and Hearst
  patterns~\cite{Hearst:1992:COLING}, distant supervision exploits the
  ``duality''~\cite{Brin:1998:WebDB} between patterns and relation
  instances; furthermore, it allows us to integrate this idea into
  DeepDive's unified probabilistic framework.  
\end{example}

\paragraph*{Learning and Inference}
In the learning and inference phase, \systemx generates a factor
graph, similar to Markov Logic, and uses techniques from
Tuffy~\cite{Niu:2011:VLDB}. The inference and learning are done using
standard techniques (Gibbs Sampling) that we describe below after
introducing the formal semantics.

\paragraph*{Error Analysis}
\systemx runs the above three phases in sequence, and at the end of
the learning and inference, it obtains a marginal probability $p$ for
each candidate fact.  To produce the final KB, the user often selects
facts in which we are highly confident, e.g., $p>0.95$. Typically, the
user needs to inspect errors and repeat, a process that we call {\em
  error analysis}. Error analysis is the process of understanding the
most common mistakes (incorrect extractions, too-specific features,
candidate mistakes, etc.) and deciding how to correct
them~\cite{Re:2014:IEEE}. To facilitate error analysis, users 
write standard SQL queries.

\subsection{Discussion of Design Choices}
\label{sec:design}
We have found three related aspects of the DeepDive approach that we
believe enable non-computer scientists to write DeepDive programs: (1)
there is no reference in a DeepDive program to the underlying machine
learning algorithms. Thus, DeepDive programs are declarative in a
strong sense. Probabilistic semantics provide a way to debug the system
independently of any algorithm. (2) DeepDive allows users to write
feature extraction code in familiar languages (Python, SQL, and
Scala). (3) \systemx fits into the familiar SQL stack, which allows
standard tools to inspect and visualize the data. A second key
property is that the user constructs an end-to-end system and then
refines the quality of the system in a pay-as-you-go
way~\cite{Madhavan:2007:CIDR}. In contrast, traditional pipeline-based
ETL scripts may lead to time and effort spent on extraction and
integration--without the ability to evaluate how important each step
is for end-to-end application quality. Anecdotally, pay-as-you-go
leads to more informed decisions about how to improve quality.

\paragraph*{Comparison with Markov Logic}
Our language is based on Markov
Logic~\cite{Niu:2011:VLDB,Domingos:2009:Book}, and our current
language inherits Markov Logic's formal semantics. However, there
are three differences in how we implement DeepDive's language:\\[-2pt]

\noindent
{\bf Weight Tying.}  As shown in rule \textsf{FE1}, \systemx allows
factors to share weights across rules, which is used in every DeepDive
system. As we will see declaring a classifier is a one-liner in
DeepDive: $Class(x) \ts R(x,f)$ with $weight=w(f)$ declares a
classifier for objects (bindings of $x$); $R(x,f)$ indicates that
object $x$ has features $f$. In standard MLNs, this would require one
rule for each feature.\footnote{Our system Tuffy introduced this
  feature to MLNs, but its semantics had not been described in the
  literature.}  In MLNs, every rule introduces a single weight,
  and the correlation structure and weight structure are
  coupled. DeepDive decouples them, which makes writing some
  applications easier.   \\[-2pt]

\noindent
{\bf User-defined Functions.} As shown in rule \textsf{FE1}, \systemx
allows the user to use user-defined functions ($\mathsf{phrase}$ in \textsf{FE1}) to
specify feature extraction rules. This allows \systemx to handle
common feature extraction idioms using regular expressions, Python
scripts, etc. This brings more of the KBC pipeline into \systemx,
which allows DeepDive to find optimization opportunities for a larger
fraction of this pipeline.\\[-2pt]

\noindent
{\bf Implication Semantics.} In the next section, we introduce a
function $g$ that counts the number of groundings in different
ways. $g$ is an example of {\em transformation
  groups}~\cite[Ch.~12]{Jaynes:2003:Book}, a technique from the
Bayesian inference literature to model different noise
distributions. Experimentally, we show that different semantics
(choices of $g$) affect the quality of KBC applications (up to 10\% in
F1 score) compared with the default semantics of MLNs. After
some notation, we give an example to illustrate how $g$ alters
the semantics.

\subsection{Semantics of a DeepDive Program}\label{sec:execution}
\begin{figure}\centering
  \includegraphics[width=0.45\textwidth]{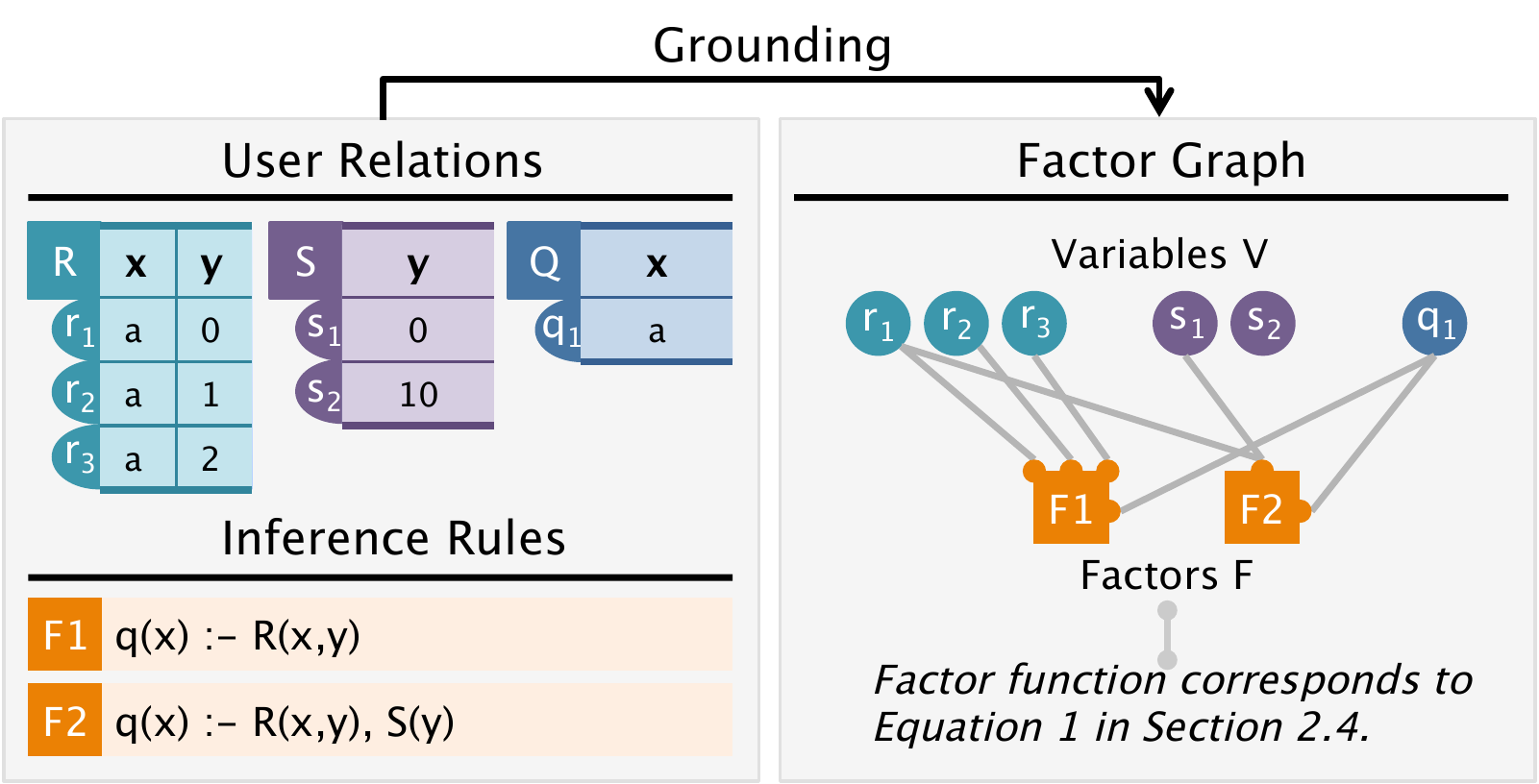}
\caption{Schematic illustration of grounding. Each tuple corresponds
  to a Boolean random variable and node in the factor graph. We create
  one factor for every set of
  groundings.}\label{fig:grounding}\end{figure}

A DeepDive program is a set of rules with weights. During inference,
the values of all weights $w$ are assumed to be known, while, in
learning, one finds the set of weights that maximizes the probability
of the evidence. As shown in Figure~\ref{fig:grounding}, a DeepDive
program defines a standard structure called a \emph{factor
  graph}~\cite{Wainwright:2008:Book}. First, we directly define the
probability distribution for rules that involve weights, as it may
help clarify our motivation. Then, we describe the corresponding
factor graph on which inference takes place.

\begin{figure}
  \begin{center}
    \begin{tabular}{cr}
\hline
Semantics & $g(n)$ \\
\hline
Linear    & $n$\\
Ratio     & $\log(1 + n)$\\
Logical   & $\mathbb{1}_{\{n > 0\}}$\\
\hline
\end{tabular}
\end{center}
\caption{Semantics for $g$ in Equation~\ref{eq:ruleweight}. }
\label{fig:semantics}
\end{figure}

Each possible tuple in the user schema--both IDB and EDB
predicates--defines a Boolean random variable (r.v.). Let $\mathbb{V}$
be the set of these r.v.'s. Some of the r.v.'s are fixed to a specific
value, e.g., as specified in a supervision rule or by training
data. Thus, $\mathbb{V}$ has two parts: a set $\mathcal{E}$ of
\emph{evidence} variables (those fixed to a specific values) and a set
$\mathcal{Q}$ of \emph{query} variables whose value the system will
infer. The class of evidence variables is further split into
\emph{positive evidence} and \emph{negative evidence}. We denote the
set of positive evidence variables as $\mathcal{P}$, and the set of
negative evidence variables as $\mathcal{N}$. An assignment to each of
the query variables yields a \emph{possible world} $I$ that must
contain all positive evidence variables, i.e., $I \supseteq
\mathcal{P}$, and must not contain any negatives, i.e., $I \cap
\mathcal{N} = \emptyset$.\\

\noindent\textbf{Boolean Rules} We first present the semantics of
\emph{Boolean} inference rules. For ease of exposition only, we assume
that there is a single domain $\D$. A rule $\gamma$ is a pair $(q,w)$
such that $q$ is a Boolean query and $w$ is a real number. An example
is as follows:
\[ q() \ts R(x,y),S(y) \quad weight=w. \]
We denote the body predicates of $q$ as $body(\bar z)$ where $\bar z$
are all variables in the body of $q()$, e.g., $\bar z = (x,y)$ in the
example above. Given a rule $\gamma=(q,w)$ and a possible world $I$,
we define the \emph{sign of $\gamma$ on $I$} as
$\mathsf{sign}(\gamma,I)= 1$ if $q() \in I$ and $-1$ otherwise.

Given $\bar{c}\in\D^{|\bar{z}|}$, a \emph{grounding} of $q$
w.r.t.~$\bar{c}$ is a substitution $body(\bar{z}/\bar{c})$, where the
variables in $\bar{z}$ are replaced with the values in $\bar{c}$. For
example, for $q$ above with $\bar c = (a,b)$ then
$body(\bar{z}/(a,b))$ yields the grounding $R(a,b), S(b)$, which is a
conjunction of facts. The \emph{support} $\mathsf{n}(\gamma,I)$ of a rule
$\gamma$ in a possible world $I$ is the number of groundings $\bar{c}$
for which $body(\bar{z}/\bar{c})$ is satisfied in $I$:
\[
	\mathsf{n}(\gamma,I)=|\{\bar{c}\in\D^{|\bar{z}|} ~:~ I\models
	body(\bar{z}/\bar{c}) ~ \}|
\]

\noindent
The \emph{weight of $\gamma$ in $I$} is the product of three terms: 
\begin{equation}\label{eq:ruleweight}
	\mathsf{w}(\gamma, I) = w \; \mathsf{sign}(\gamma,I) \; g(\mathsf{n}(\gamma,
	I)),
\end{equation}
where $g$ is a real-valued function defined on the natural
numbers. For intuition, if $w(\gamma,I) > 0$, it adds a weight that
indicates that the world is more likely. If $w(\gamma,I) < 0$, it
indicates that the world is less likely. As motivated above, we
introduce $g$ to support multiple
semantics. Figure~\ref{fig:semantics} shows choices for $g$ that are
supported by DeepDive, which we compare in an example below.

Let $\Gamma$ be a set of Boolean rules, the weight of $\Gamma$ on a
possible world $I$ is defined as
\[
\mathsf{W}(\Gamma,I)=\sum_{\gamma\in\Gamma}\mathsf{w}(\gamma,I).
\]
This function allow us to define a probability distribution over the set $J$ of
possible worlds:
\begin{equation}\label{eq:probability}
	\Pr[I]=\mathsf{Z}^{-1}\mathrm{exp}(\mathsf{W}(\Gamma,I)) \text{ where }
\mathsf{Z}=\sum_{I\in J} \mathrm{exp}(\mathsf{W}(\Gamma,I)), 
\end{equation}
and $\mathsf{Z}$ is called the \emph{partition function}. This
framework is able to compactly specify much more sophisticated
distributions than traditional probabilistic databases~\cite{Suciu:2011:Book}.  

\begin{example} We illustrate the semantics by example.
 From the Web, we could extract a set of relation mentions that
 supports {\em ``Barack Obama is born in Hawaii''} and another set of
 relation mentions that support {\em ``Barack Obama is born in
   Kenya.''}  These relation mentions provide conflicting information,
 and one common approach is to ``vote.'' We abstract this as up or
 down votes about a fact $q()$.
\begin{align*}
  q() & \mathop{\ts} \;Up(x) \quad & weight = 1. \\
  q() & \mathop{\ts} \;Down(x) \quad & weight = -1.
\end{align*}
We can think of this as a having a single random variable $q()$ in
which the size of $Up$ (resp. $Down$) is an evidence relation that
indicates the number of ``Up'' (resp. ``Down'') votes. There are only
two possible worlds: one in which $q() \in I$ (is true) and not. Let
$|Up|$ and $|Down|$ be the sizes of $Up$ and $Down$. Following
Equation~\ref{eq:ruleweight} and \ref{eq:probability}, we have
\[ \Pr[q()] = \frac{e^{W}}{e^{-W} + e^{W}} \]
where
\[ W = g(|Up|) - g(|Down|). \]
Consider the case when $|Up| = 10^6$ and $|Down| = 10^6 - 100$. In
some scenarios, this small number of differing votes could be due to
random noise in the data collection processes. One would expect a
probability for $q()$ close to $0.5$. In the linear semantics
$g(n)=n$, the probability of $q$ is $(1 + e^{-200})^{-1} \approx 1 -
e^{-200}$, which is extremely close to $1$. In contrast, if we set
$g(n) = \log (1 + n)$, then $\Pr[q()] \approx 0.5$. Intuitively, the
probability depends on their {\em ratio} of these votes. The {\em
  logical} semantics $g(n) = \mathbb{1}_{n > 0}$ gives exactly
$\Pr[q()] = 0.5$. However, it would do the same if $|Down| = 1$. Thus,
logical semantics may ignore the strength of the voting
information. At a high level, ratio semantics can learn weights from
examples with different raw counts but similar ratios. In contrast,
linear is appropriate when the raw counts themselves are meaningful.
\label{ex:voting}\end{example}

No semantic subsumes the other, and each is appropriate in some
application. We have found that in many cases the ratio semantics is
more suitable for the application that the user wants to model.  We
show in the full version that these semantics also affect efficiency
empirically and theoretically--even for the above simple
example.  Intuitively, sampling converges faster in the logical or
ratio semantics because the distribution is less sharply peaked, which
means that the sampler is less likely to get stuck in local minima.\\

\noindent\textbf{Extension to General Rules.} Consider a general inference
rule $\gamma=(q,w)$, written as:
\[ q(\bar{y})\ts body(\bar{z}) \quad \text{  weight}=w(\bar{x}). \] where
$\bar{x}\subseteq\bar{z}$ and $\bar{y}\subseteq\bar{z}$. This
extension allows {\em weight tying}. Given
$\bar{b}\in\D^{|\bar{x}\cup\bar{y}|}$ where $\bar{b}_x$
(resp. $\bar{b}_y$) are the values of $\bar b$ in $\bar{x}$
(resp. $\bar{y}$), we expand $\gamma$ to a set $\Gamma$ of Boolean rules
by substituting $\bar x \cup \bar y$ with values from $\D$ in all
possible ways.
\begin{align*}
	\Gamma=\{(q_{\bar{b}_y},w_{\bar{b}_x}) ~|~ 
	q_{\bar{b}_y}()\ts body(\bar{z}/\bar{b}) \mbox{ and }
	w_{\bar{b}_x}=w(\bar{x}/\bar{b}_x)\}
\end{align*}
where each $q_{\bar{b}_y}()$ is a \emph{fresh} symbol for distinct
values of $\bar b_{t}$, and $w_{\bar{b}_x}$ is a real number. Rules
created this way may have free variables in their bodies, e.g., $q(x)
\ts R(x,y,z)$ with $w(y)$ create $|\D|^2$ different rules of the form
$q_{a}() \ts R(a,b,z)$, one for each $(a,b) \in \D^2$, and rules
created with the same value of $b$ share the same weight. Tying
weights allows one to create models succinctly.

\begin{example} \label{exp:lr}
We use the following as an example:
  \[ Class(x) \ts R(x,f) \quad weight = w(f). \]
  This declares a binary classifier as follows. Each binding for $x$
  is an object to classify as in $Class$ or not. The relation $R$
  associates each object to its features. E.g., $R(a,f)$ indicates
  that object $a$ has a feature $f$. $weight = w(f)$
  indicates that weights are functions of feature $f$; thus,
  the same weights are tied across values for $a$. This rule 
  declares a logistic regression classifier.
\end{example}

It is straightforward formal extension to let weights be functions of
the return values of UDFs as we do in DeepDive.

\subsection{Inference on Factor Graphs}

As in Figure~\ref{fig:grounding}, DeepDive explicitly constructs a
factor graph for inference and learning using a set of SQL
queries. Recall that a factor graph is a triple $(V,F,\hat{w})$ in
which $V$ is a set of nodes that correspond to Boolean random
variables, $F$ is a set of hyperedges (for $f \in F$, $f \subseteq
V$), and $\hat{w} : F \times \{0,1\}^{V} \to \mathbb{R}$ is a weight
function. We can identify possible worlds with assignments since each
node corresponds to a tuple; moreover, in DeepDive, each hyperedge $f$
corresponds to the set of groundings for a rule $\gamma$. In DeepDive,
$V$ and $F$ are explicitly created using a set of SQL queries. These
data structures are then passed to the sampler, which runs outside the
database, to estimate the marginal probability of each node or tuple
in the database. Each tuple is then reloaded into the database with
its marginal probability.

\begin{example}
  Take the database instances and rules in Figure~\ref{fig:grounding}
  as an example, each tuple in relation $R$, $S$, and $Q$ is
  a random variable, and $V$ contains all random variables. The
  inference rules $F1$ and $F2$ ground factors with the same name
  in the factor graph as illustrated in Figure~\ref{fig:grounding}.
  Both $F1$ and $F2$ are implemented as SQL in DeepDive.
\end{example}

To define the semantics, we use Equation~\ref{eq:ruleweight} to define
$\hat{w}(f,I) = w(\gamma,I)$, in which $\gamma$ is the rule
corresponding to $f$. As before, we define $\hat{W}(F,I) = \sum_{f \in
  F} \hat{w}(f,I)$, and then the probability of a possible world is
the following function:
\[ \Pr[ I ] = Z^{-1} \exp\left\{ \hat{W}(F,I) \right\} \text{ where } Z = \sum_{I \in J} \exp\{ \hat{W}(F,I)\}  \]
\noindent
The main task that DeepDive conducts on factor graphs is statistical
inference, i.e., for a given node, what is the marginal probability that this
node takes the value $1$? Since a node takes value $1$ when a tuple is
in the output, this process computes the marginal probability values
returned to users. In general, computing these marginal probabilities
is $\sharp\mathsf{P}$-hard~\cite{Wainwright:2008:Book}. Like many
other systems, DeepDive uses Gibbs sampling~\cite{Robert:2005:Book} to
estimate the marginal probability of every tuple in the database.


\section{Incremental KBC} \label{sec:tradeoff}

To help the KBC system developer be more efficient, we developed
techniques to incrementally perform the grounding and inference step
of KBC execution.

\paragraph*{Problem Setting}
Our approach to incrementally maintaining a KBC system runs in two
phases. {\bf (1) Incremental Grounding.} The goal of the incremental
grounding phase is to evaluate an update of the DeepDive program to
produce the ``delta'' of the modified factor graph, i.e., the modified
variables $\Delta V$ and factors $\Delta F$. This phase consists of
relational operations, and we apply classic incremental view
maintenance techniques. {\bf (2) Incremental Inference.} The goal of
incremental inference is given $(\Delta V, \Delta F)$ run statistical
inference on the changed factor graph.

\subsection{Standard Techniques: Delta Rules} \label{sec:standard}

Because \systemx is based on SQL, we are able to take advantage of
decades of work on incremental view maintenance. The input to this
phase is the same as the input to the grounding phase, a set of SQL
queries and the user schema. The output of this phase is how the
output of grounding changes, i.e., a set of modified variables $\Delta
V$ and their factors $\Delta F$. Since $V$ and $F$ are simply views
over the database, any view maintenance techniques can be applied to
incremental grounding. DeepDive uses \textsc{DRed}
  algorithm~\cite{Gupta:1993:SIGMOD} that handles both additions and
  deletions. Recall that in \textsc{DRed}, for each relation $R_i$ in
the user's schema, we create a {\em delta relation}, $R^{\delta}_i$,
with the same schema as $R_i$ and an additional column $count$. For
each tuple $t$, $t.count$ represents the number of derivations of $t$
in $R_i$.  On an update, \systemx updates delta relations in two
steps. First, for tuples in $R^{\delta}_i$, \systemx directly updates
the corresponding counts. Second, a SQL query called a {\it ``delta
  rule''}\footnote{ For example, for the grounding procedure
    illustrated in Figure~\ref{fig:grounding}, the delta rule for $F1$
    is $q^{\delta}(x) :- R^{\delta}(x,y)$.}  is executed which
processes these counts to generate modified variables $\Delta V$ and
factors $\Delta F$.  We found that the overhead \textsc{DRed} is
modest and the gains may be substantial, and so DeepDive always runs
\textsc{DRed}--except on initial load.

\subsection{Novel Techniques for Incremental Maintenance of Inference}
\label{sec:inctech}

We present three techniques for the incremental inference phase on
factor graphs: given the set of modified variables $\Delta V$ and
modified factors $\Delta F$ produced in the incremental grounding
phase, our goal is to compute the new distribution. We split the
problem into two phases. In the {\em materialization phase}, we are
given access to the entire DeepDive program, and we attempt to store
information about the original distribution, denoted $\Pr^{(0)}$. Each
approach will store different information to use in the next phase,
called the {\em inference phase}. The input to the inference phase is
the materialized data from the preceding phase and the changes made to
the factor graph, the modified variables $\Delta V$ and factors
$\Delta F$. Our goal is to perform inference with respect to the
changed distribution, denoted $\Pr^{(\Delta)}$. For each approach, we
study its space and time costs for materialization and the time cost
for inference. We also analyze the empirical tradeoff between the
approaches in Section~\ref{sec:tradeoff:discussion}.

\subsubsection{Strawman: Complete Materialization}
The strawman approach, complete materialization, is computationally
expensive and often infeasible. We use it to set a baseline for 
other approaches.

\textbf{Materialization Phase} We explicitly store the value of the
probability $\Pr[I]$ for every possible world $I$. This approach has
perfect fidelity, but storing all possible worlds takes an exponential
amount of space and time in the number of variables in the original
factor graph. Thus, the strawman approach is often infeasible on even
moderate-sized graphs.\footnote{Compared with running inference from
  scratch, the strawman approach does not materialize any
  factors. Therefore, it is necessary for strawman to enumerate each
  possible world and save their probability because we do not know
  {\em a priori} which possible world will be used in the inference
  phase.}

\textbf{Inference Phase} 
We use Gibbs sampling: even if the distribution has changed to
$\Pr^{(\Delta)}$, we only need access to the new factors in
$\Delta \Pi_\mathcal{F}$ and to $\Pr[I]$ to perform the Gibbs update. The
speed improvement arises from the fact that we do not need to access all
factors from the original graph and perform a computation with them, since we
can look them up in $\Pr[I]$.

\subsubsection{Sampling Approach} \label{sec:sampling}
The sampling approach is a standard technique to improve over
the strawman approach by storing a set of possible worlds sampled from
the original distribution instead of storing all possible
worlds. However, as the updated distribution $\Pr^{(\Delta)}$ is different from
the distribution used to draw the stored samples, we cannot reuse them
directly. We use a (standard) Metropolis-Hastings scheme to ensure convergence
to the updated distribution.

\textbf{Materialization Phase} In the materialization phase, we store
a set of possible worlds drawn from the original distribution. For
each variable, we store the set of samples as a {\em tuple bundle}, as
in MCDB~\cite{Jampani:2008:SIGMOD}. A single sample for one random
variable only requires 1 bit of storage. Therefore, the sampling
approach can be efficient in terms of materialization space. In the
KBC systems we evaluated, $100$ samples require less than 5\% of the
space of the original factor graph.

\textbf{Inference Phase} We use the samples to generate {\em
  proposals} and adapt them to estimate the up-to-date
distribution. This idea of using samples from similar distributions as
proposals is standard in statistics, e.g., importance sampling,
rejection sampling, and different variants of Metropolis{-\penalty0\hskip0pt\relax}Hastings
methods. After investigating these approaches, in \systemx, we use the
{\em independent Metropolis-Hastings}
approach~\cite{Andrieu:2003:ML,Robert:2005:Book}, which generates
proposal samples and accepts these samples with an {\em acceptance
  test}. We choose this method only because the acceptance test can be
evaluated using the sample, $\Delta V$, and $\Delta F$--without the
entire factor graph. Thus, we may fetch many fewer factors than in the
original graph, but we still converge to the correct answer.

The fraction of accepted samples is called the {\em acceptance rate},
and it is a key parameter in the efficiency of this approach. The
approach may exhaust the stored samples, in which case the method
resorts to another evaluation method or generates fresh samples. 

\begin{algorithm}[t]
\scriptsize
\begin{algorithmic}[1]
\Require Factor graph $FG = (V, F)$,
regularization parameter $\lambda$, number of samples $N$ for approximation.
\Ensure An approximated factor graph $FG' =
(V, F')$
\algrule

\State $I_1,...,I_N \leftarrow N$ samples drawn from  $FG$.

\State $NZ \leftarrow $ \{$(v_i, v_j)$: $v_i$
and $v_j$ are in some factor in $FG$\}.

\State $M \leftarrow $ covariance matrix estimated using $I_1,...,I_N$, such
that $M_{ij}$ is the covariance between variable $i$ and variable $j$. Set
$M_{ij} = 0$ if $(v_i, v_j) \not\in NZ$.

\State Solve the following optimization problem using gradient
descent~\cite{Banerjee:2008:JMLR}, and let the result be $\hat{X}$
\begin{eqnarray*}
~ & \arg\max_X &\log \det {X}\\
~ & s.t., & X_{kk}=M_{kk} + 1/3,\\
~ &       & |X_{kj}-M_{kj}| \leq \lambda\\
~ &       & X_{kj} = 0~if~(v_k, v_j) \not \in NZ
\end{eqnarray*}

\ForAll{$i,j$ s.t. $\hat{X}_{ij} \ne 0$}
\State Add in $F'$ a factor from $(v_i, v_j)$ with
weight $\hat{X}_{ij}$.

\EndFor

\State {\bf return} $FG' = (V, F')$.

\end{algorithmic}
\caption{Variational Approach (Materialization)}
\label{alg:var}
\end{algorithm}

\begin{figure*}[t]
\centering
\scriptsize
\includegraphics[width=0.98\textwidth]{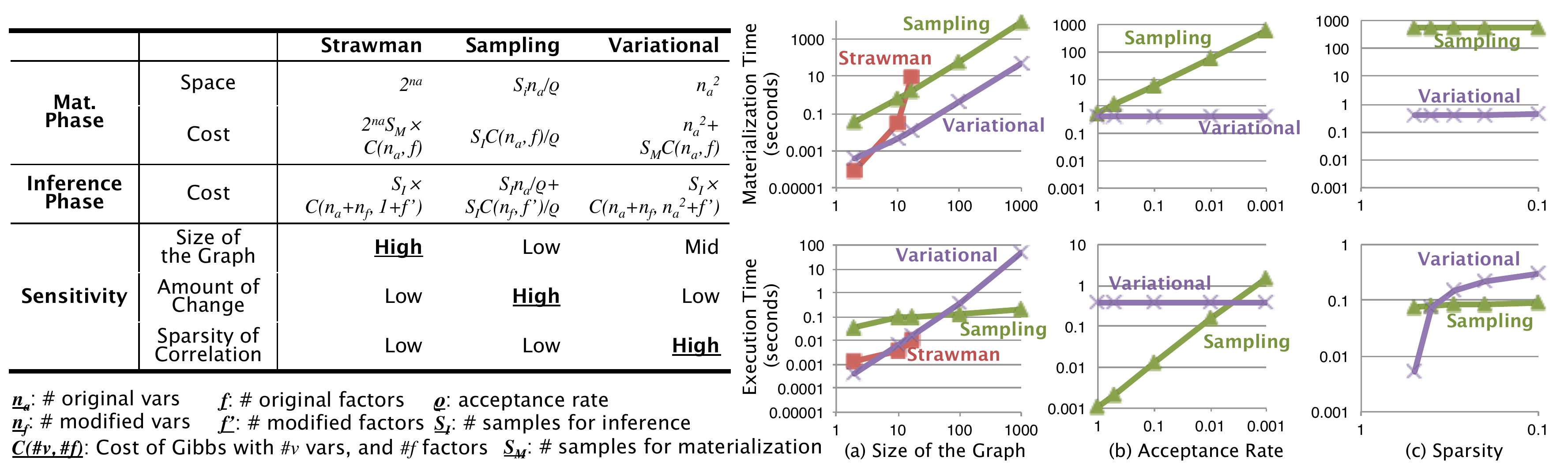}
\caption{A Summary of the tradeoffs. Left: An analytical cost model
  for different approaches; Right: Empirical examples that illustrate
  the tradeoff space. All converge to $<$0.1\% loss, and
  thus, have comparable quality.}
\label{fig:cost}
\end{figure*}

\subsubsection{Variational Approach}
The intuition behind our variational approach is as follows: rather
than storing the exact original distribution, we store a factor graph
with \emph{fewer factors} that approximates the original
distribution. On the smaller graph, running inference and learning is
often faster.

\textbf{Materialization Phase} The key idea of the variational
approach is to approximate the distribution using simpler or sparser
correlations. To learn a sparser model, we use Algorithm~\ref{alg:var}
which is a log-determinant relaxation~\cite{Wainwright:2006:TSP} with
a $\ell_1$ penalty term~\cite{Banerjee:2008:JMLR}. We want to
understand its strengths and limitations on KBC problems, which is
novel. This approach uses standard techniques for learning that are
already implemented in DeepDive~\cite{Zhang:2014:VLDB}.

The input is the original factor graph and two parameters: the number
of samples $N$ to use for approximating the covariance matrix, and the
regularization parameter $\lambda$, which controls the sparsity of the
approximation. The output is a new factor graph that has only binary
potentials. The intuition for this procedure comes from graphical
model structure learning: an entry $(i,j)$ is present in the inverse
covariance matrix only if variables $i$ and $j$ are connected in the
factor graph. Given these inputs, the algorithm first draws a set of
$N$ possible worlds by running Gibbs sampling on the original factor
graph. It then estimates the covariance matrix based on these samples
(Lines 1-3). Using the estimated covariance matrix, our algorithms
solves the optimization problem in Line 4 to estimate the inverse
covariance matrix $\hat{X}$. Then, the algorithm creates one factor
for each pair of variables such that the corresponding entry in
$\hat{X}$ is non-zero, using the value in $\hat{X}$ as the new weight
(Line 5-7). These are all the factors of the approximated factor graph
(Line 8).

\textbf{Inference Phase} Given an update to the factor graph (e.g.,
new variables or new factors), we simply apply this update to the
approximated graph, and run inference and learning directly on the
resulting factor graph. As shown in Figure~\ref{fig:cost}(c), the
execution time of the variational approach is roughly linear in the
sparsity of the approximated factor graph.  Indeed, the execution time
of running statistical inference using Gibbs sampling is dominated by
the time needed to fetch the factors for each random variable, which
is an expensive operation requiring random access. Therefore, as the
approximated graph becomes sparser, the number of factors decreases and
so does the running time.

\textbf{Parameter Tuning} We are among the first to use these methods
in KBC applications, and there is little literature about tuning
$\lambda$. Intuitively, the smaller $\lambda$ is, the better the
approximation is--but the less sparse the approximation is. To
understand the impact of $\lambda$ on quality, we show in
Figure~\ref{fig:knee} the quality F1 score of a DeepDive program
called News (see Section~\ref{sec:exp}) as we vary the regularization
parameter. As long as the regularization parameter $\lambda$ is small
(e.g., less than 0.1), the quality does not change significantly. In
all of our applications we observe that there is a relatively large
``safe'' region from which to choose $\lambda$.  In fact, for all five
systems in Section~\ref{sec:exp}, even if we set $\lambda$ at $0.1$ or
0.01, the impact on quality is minimal (within 1\%), while the impact
on speed is significant (up to an order of magnitude). Based on
Figure~\ref{fig:knee}, \systemx supports a simple search protocol to
set $\lambda$. We start with a small $\lambda$, e.g., $0.001$, and
increase it by 10$\times$ until the KL-divergence is larger than a
user-specified threshold, specified as a parameter in \systemx.

\subsubsection{Tradeoffs} \label{sec:tradeoff:discussion}

We studied the tradeoff between different approaches and summarize the
empirical results of our study in Figure~\ref{fig:cost}. The
performance of different approaches may differ by more than two
orders of magnitude, and neither of them dominates the other. We use
a synthetic factor graph with pairwise factors\footnote{In
  Figure~\ref{fig:cost}, the numbers are reported for a factor graph
  whose factor weights are sampled at random from $[-0.5,0.5]$. We
  also experimented with different intervals ($[-0.1,0.1]$, $[-1,1]$,
  $[-10,10]$), but these had no impact on the tradeoff} and control the
following axes:

\begin{description}
\compactify
\item {\bf (1) Number of variables} in the factor graph. In our experiments, we
	set the number of variables to values in \{2, 10, 17, 100, 1000,
	10000\}.
         
      \item {\bf (2) Amount of change}. How much the distribution
        changes affects efficiency, which manifests itself in the
        acceptance rate: the smaller the acceptance rate is, the more
        difference there will be in the distribution. We set the
        acceptance rate to values in \{1.0, 0.5, 0.1, 0.01\}.

      \item {\bf (3) Sparsity of correlations}. This is the ratio between the number
	of non-zero weights and the total weight. We set the sparsity to
	values in $\{0.1,0.2,0.3,0.4,0.5,1.0\}$ by selecting uniformly at random a
	subset of factors and set their weight to zero.
\end{description}

We now discuss the results of our exploration of the tradeoff space, presented
in Figure~\ref{fig:cost}(a-c).\\

\noindent
{\bf Size of the Factor Graph} Since the materialization cost of the
strawman is exponential in the size of the factor graph, we observe
that, for graphs with more than 20 variables, the strawman is
significantly slower than either the sampling approach or the variational
approach.  Factor graphs arising from KBC systems usually contain a
much larger number of variables; therefore, from now on, we focus on
the tradeoff between sampling and variational approaches.\\[-2pt]

\noindent
{\bf Amount of Change} As shown in Figure~\ref{fig:cost}(b), when the
acceptance rate is high, the sampling approach could outperform the
variational one by more than two orders of magnitude. When the
acceptance rate is high, the sampling approach requires no computation
and so is much faster than Gibbs sampling. In contrast, when the
acceptance rate is low, e.g., 0.1\%, the variational approach could be
more than 5$\times$ faster than the sampling approach. An acceptance
rate lower than $0.1$\% occurs for KBC operations when one updates the
training data, adds many new features, or concept drift
happens during the development of KBC systems.\\[-2pt]

\noindent
{\bf Sparsity of Correlations} As shown in Figure~\ref{fig:cost}(c),
when the original factor graph is sparse, the variational approach can
be 11$\times$ faster than the sampling approach. This is because the
approximate factor graph contains less than 10\% of the factors than
the original graph, and it is therefore much faster to run inference
on the approximate graph. On the other hand, if the original factor
graph is too dense, the variational approach could be more than
7$\times$ slower than the sampling one, as it is essentially
performing inference on a factor graph with a size similar to that of
the original graph.\\[-2pt]

\noindent
{\bf Discussion: Theoretical Guarantees} We discuss the
  theoretical guarantee that each materialization strategy
  provides. Each materialization method inherits the guarantee of that
  inference technique. The strawman approach retains the same
  guarantees as Gibbs sampling; For the sampling approach use standard
  Metropolis-Hasting scheme. Given enough time, this approach will
  converge to the true distribution. For the variational approach, the
  guarantees are more subtle and we point the reader to the
  consistency of structure estimation of Gaussian Markov random
  field~\cite{Ravikumar:2018:NIPS} and log-determinate
  relaxation~\cite{Wainwright:2006:TSP}. These results are
  theoretically incomparable, motivating our empirical study. 

\begin{figure}[t]
\centering
\scriptsize
\includegraphics[width=0.4\textwidth]{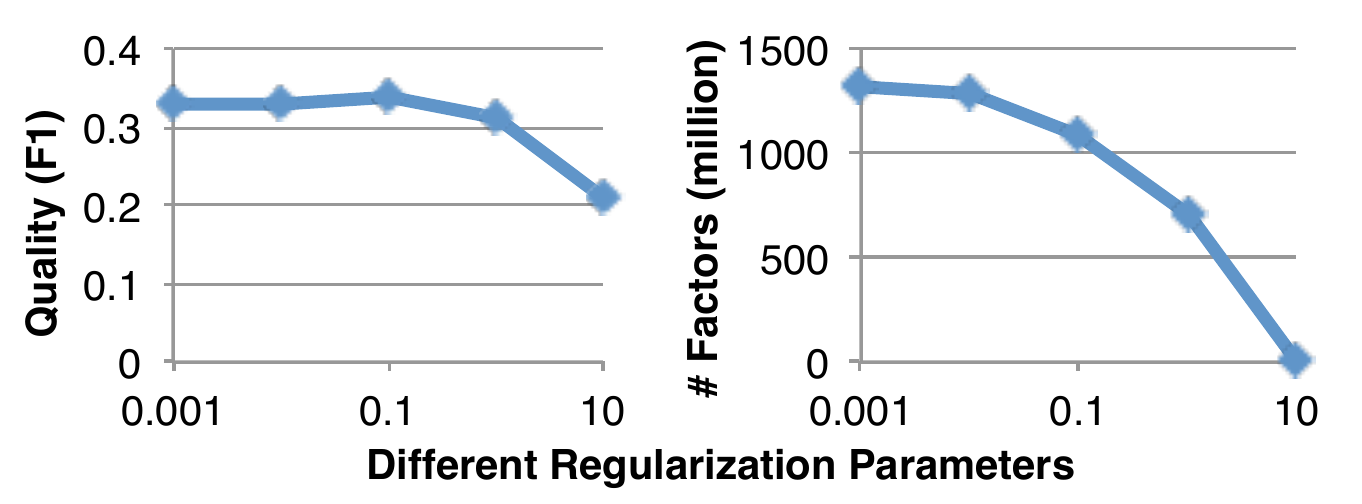}
\caption{Quality and number of factors of the News corpus with
  different regularization parameters for the variational approach.}
\label{fig:knee}
\end{figure}


\subsection{Choosing Between Different Approaches}

From the study of the tradeoff space, neither the sampling approach
nor the variational approach dominates the other, and their relative
performance depends on how they are being used in KBC. We propose to
materialize the factor graph using both the sampling approach {\em
  and} the variational approach, and defer the decision to the
inference phase when we can observe the workload.

\textbf{Materialization Phase} Both approaches need samples from the
original factor graph, and this is the dominant cost during
materialization. A key question is {\em ``How many samples should we
  collect?''}  We experimented with several heuristic methods to
estimate the number of samples that are needed, which requires
understanding how likely future changes are, statistical
considerations, etc. These approaches were difficult for users to
understand, so \systemx takes a best-effort approach: it generates as
many samples as possible when idle or within a user-specified time
interval.

\textbf{Inference Phase} Based on the tradeoffs analysis, we
developed a \emph{rule-based optimizer} with the following set of rules:
\begin{itemize}
\compactify
	\item If an update does not change the structure of the graph,
          choose the sampling approach.
	
\item If an update modifies the evidence, choose the variational
  approach.

	\item If an update introduces new features, choose the
          sampling approach.

          \item Finally, if we run out of samples, use
            the variational approach.
\end{itemize}

\noindent
This simple set of rules is used in our
experiments.

\begin{figure}
\scriptsize
\centering
\begin{tabular}{c|rrr|rr}
\hline
{\bf System} & \# Docs & \# Rels & \# Rules & \# vars & \# factors\\
\hline
Adversarial        & 5M     & 1             & 10        & 0.1B & 0.4B     \\
News               & 1.8M    &  34          & 22       & 0.2B & 1.2B\\ 
Genomics           & 0.2M    & 3            & 15      & 0.02B  & 0.1B     \\
Pharma.   & 0.6M    & 9            & 24      & 0.2B & 1.2B\\
Paleontology       & 0.3M    &  8           & 29       & 0.3B & 0.4B\\
\hline
\end{tabular}
\caption{Statistics of KBC systems we used in experiments. 
The \# vars and \# factors are for factor graphs that
contain all rules.}
\label{tab:datasets}
\end{figure}

\begin{figure}
\scriptsize
\centering
\begin{tabular}{c|l}
\hline
{\bf Rule} & {\bf Description}\\
\hline
\multirow{2}{*}{\bf A1} & Calculate marginal probability for variables\\
                        & 
or variable pairs.\\
\hline
{\bf FE1} & Shallow NLP features (e.g., word sequence)\\
{\bf FE2} & Deeper NLP features (e.g., dependency path)\\
\hline
{\bf I1} & Inference rules (e.g., symmetrical \textsf{HasSpouse}).\\
\hline
{\bf S1} & Positive examples\\
{\bf S2} & Negative examples\\
\hline
\end{tabular}
\caption{The set of rules in News. See Section~\ref{sec:expsetting}}
\label{tab:rules}
\end{figure}

\section{Experiments} \label{sec:exp}
We conducted an experimental evaluation of \systemx for
incremental maintenance of KBC systems.

\subsection{Experimental Settings} \label{sec:expsetting}
To evaluate \systemx, we used DeepDive programs developed by our users
over the last three years from paleontologists, geologists,
biologists, a defense contractor, and a KBC competition.  These are
high-quality KBC systems: two of our KBC systems for natural sciences
achieved quality comparable to (and sometimes better than) human
experts, as assessed by double-blind experiments, and our KBC system
for a KBC competition is the top system among all 45 submissions from
18 teams as assessed by professional annotators. To simulate the
development process, we took snapshots of DeepDive programs at the end
of every development iteration, and we use this dataset of snapshots
in the experiments to understand our hypothesis that incremental
techniques can be used to improve development speed.\\[-2pt]

\begin{figure*}[t]
\scriptsize
\centering
\begin{tabular}{c|rrr|rrr|rrr|rrr|rrr}
\hline
\multirow{2}{*}{\bf Rule}
    & \multicolumn{3}{c|}{\bf Adversarial} & \multicolumn{3}{c|}{\bf News} & \multicolumn{3}{c|}{\bf Genomics}  & \multicolumn{3}{c|}{\bf Pharma.} & \multicolumn{3}{c}{\bf Paleontology} \\
& Rerun  & Inc.  & $\times$ & Rerun  & Inc.  & $\times$ & Rerun  & Inc.  & $\times$ & Rerun  & Inc.  & $\times$ & Rerun  & Inc.  & $\times$\\
\hline
{\bf A1}  & 1.0 &0.03 &33$\times$ &2.2  &0.02 &112$\times$  &0.3  &0.01 &30$\times$ &3.6  &0.11 &33$\times$ &2.8  &0.3  &10$\times$     \\
\hline
{\bf FE1} & 1.1 &0.2  &7$\times$  &2.7  &0.3  &10$\times$ &0.4  &0.07 &6$\times$  &3.8  &0.3  &12$\times$ &3.0  &0.4  &7$\times$ \\
{\bf FE2} & 1.2 &0.2  &6$\times$ &3.0 &0.3  &10$\times$ &0.4  &0.07 &6$\times$  &4.2  &0.3  &12$\times$ &3.3  &0.4  &8$\times$ \\
\hline
{\bf I1}  &1.3  &0.2  &6$\times$ &3.6 &0.3  &10$\times$ &0.5  &0.09 &6$\times$    &4.4  &1.4  &3$\times$  &3.8  &0.5  &8$\times$     \\
\hline
{\bf S1}  &1.3  &0.2  &6$\times$ &3.6 &0.4  &8$\times$  &0.6  &0.1  &6$\times$    &4.7  &1.7  &3$\times$  &4.0  &0.5  &7$\times$ \\ 
{\bf S2}  &1.3  &0.3  &5$\times$ &3.6 &0.5  &7$\times$  &0.7  &0.1  &7$\times$  &4.8  &2.3  &3$\times$  &4.1  &0.6  &7$\times$ \\
\hline
\end{tabular}
\caption{End-to-end efficiency of incremental inference and learning. All execution times are in hours. The column $\times$ refers to the speedup of \textsc{Incremental} (Inc.) over \textsc{Rerun}.}
\label{fig:e2e}
\end{figure*}

\noindent
\textbf{Datasets and Workloads} To study the efficiency of \systemx,
we selected five KBC systems, namely (1) News, (2) Genomics, (3)
Adversarial, (4) Pharmacogenomics, and (5) Paleontology. Their names
refers to the specific domains on which they focus.
Figure~\ref{tab:datasets} illustrates the statistics of these KBC
systems and of their input datasets. We group all rules in each system
into six rule templates with four workload categories. We focus on the
News system below.

The News system builds a knowledge base between persons, locations,
and organizations, and contains 34 different relations, e.g.,
\textsf{HasSpouse} or \textsf{MemberOf}. The input to the KBC system
is a corpus that contains 1.8 million news articles and Web pages. We
use four types of rules in News in our experiments, as shown in
Figure~\ref{tab:rules}, error analysis (rule A1), candidate generation
and feature extraction (FE1, FE2), supervision (S1, S2), and inference
(I1), corresponding to the steps where these rules are used.

Other applications are different in terms of the quality of the
text. We choose these systems as they span a large range in the
spectrum of quality: Adversarial contains advertisements collected
from websites where each document may have only 1-2 sentences with
grammatical errors; in contrast, Paleontology contains well-curated
journal articles with precise, unambiguous writing and simple
relationships. Genomics and Pharma have precise texts, but the goal is
to extract relationships that are more linguistically ambiguous
compared to the Paleontology text. News has slightly degraded writing
and ambiguous relationships, e.g., ``member of.''  Rules with the same
prefix, e.g., FE1 and FE2, belong to the same category, e.g., feature
extraction.\\[-2pt]

\noindent
\textbf{DeepDive Details}
\systemx is implemented in Scala and C++, and we use Greenplum to
handle all SQL. All feature extractors are written in Python. The
statistical inference and learning and the incremental maintenance
component are all written in C++. All experiments are run on a machine
with four CPUs (each CPU is a 12-core 2.40 GHz Xeon E5-4657L), 1 TB
RAM, and 12$\times$1TB hard drives and running Ubuntu 12.04. For these
experiments, we compiled \systemx with Scala 2.11.2, g++-4.9.0 with
-O3 optimization, and Python 2.7.3. In Genomics and Adversarial,
Python 3.4.0 is used for feature extractors.

\begin{figure}[t]
\centering
\includegraphics[width=\columnwidth]{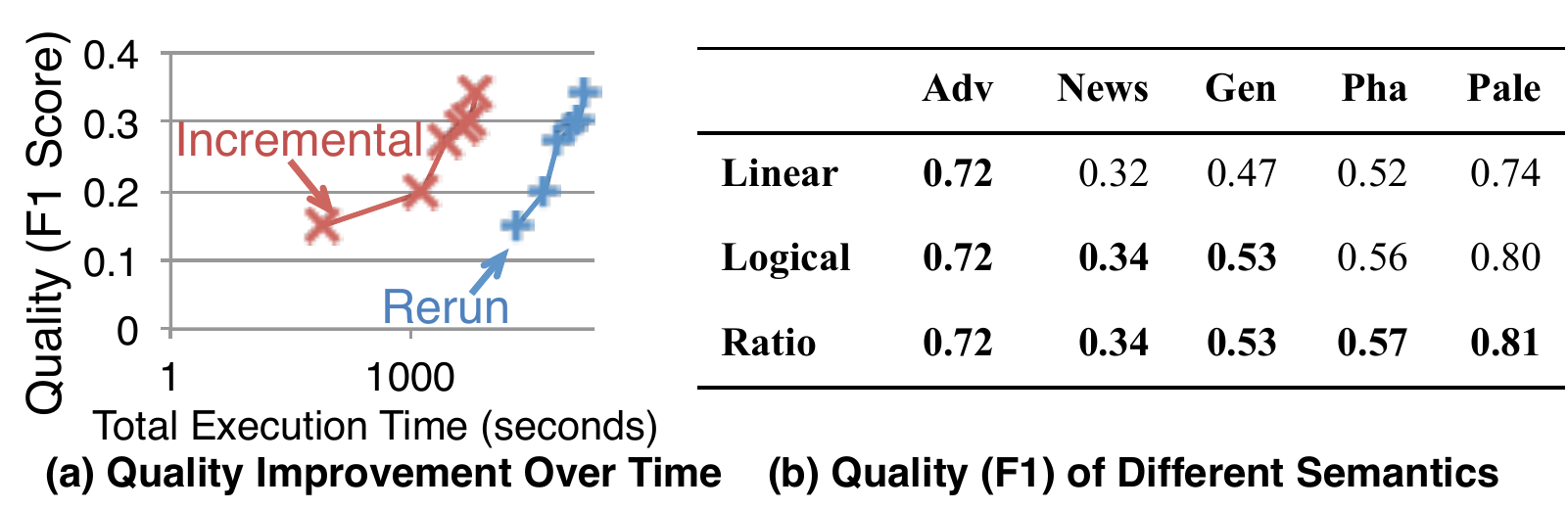}
\caption{(a) Quality improvement over time; (b) Quality for different semantics.}
\label{fig:quality}
\end{figure}

\subsection{End-to-end Performance and Quality}
We built a modified version of \systemx called \textsc{Rerun}, which
given an update on the KBC system, runs the DeepDive program from
scratch. \systemx, which uses all techniques, is called
\textsc{Incremental}. The results of our evaluation show that \systemx
is able to speed up the development of high-quality KBC systems
through incremental maintenance with little impact on quality. We set
the number of samples to collect during execution to \{10, 100, 1000\}
and the number of samples to collect during materialization to
\{1000, 2000\}. We report results for $(1000,2000)$, as results for
other combinations of parameters are similar. \\[-2pt]

\noindent
\textbf{Quality Over Time} \label{sec:exp:quality}
We first compare \textsc{Rerun} and \textsc{Incremental} in terms of
the wait time that developers experience to improve the quality of a
KBC system. We focus on News because it is a well-known benchmark
competition. We run all six rules sequentially for both \textsc{Rerun}
and \textsc{Incremental}, and after executing each rule, we report the
quality of the system measured by the F1 score and the {\em
  cumulative} execution time. Materialization in the
\textsc{Incremental} system is performed only
once. Figure~\ref{fig:quality}(a) shows the results. Using
\textsc{Incremental} takes significantly less time than
\textsc{Rerun} to achieve the same quality. To achieve an F1 score of
0.36 (a competition-winning score), \textsc{Incremental} is 22$\times$
faster than \textsc{Rerun}. Indeed, each run of \textsc{Rerun} takes
$\approx6$ hours, while a run of \textsc{Incremental} takes at most 30
minutes.

We further compare the facts extracted by \textsc{Incremental} and
\textsc{Rerun} and find that these two systems not only have similar
end-to-end quality, but are also similar enough to support common
debugging tasks. We examine the facts with high-confidence in
\textsc{Rerun} ($>0.9$ probability), 99\% of them also appear in
\textsc{Incremental}, and vice versa. High confidence extractions are
used by the developer to debug precision issues. Among all
facts, we find that at most 4\% of them have a probability that
differs by more than $0.05$.  The similarity between snapshots
suggests, our incremental maintenance techniques can be used for
debugging.\\[-2pt]

\noindent
\textbf{Efficiency of Evaluating Updates}
We now compare \textsc{Rerun} and \textsc{Incremental} in terms of
their speed in evaluating a given update to the KBC system. To better
understand the impact of our technical contribution, we divide the
total execution time into parts: (1) the time used for feature
extraction and grounding; and (2) the time used for statistical
inference and learning. We implemented classical incremental
materialization techniques for feature extraction and grounding, which
achieves up to a 360$\times$ speedup for rule FE1
in News. We get this speedup for
free using standard RDBMS techniques, a key design decision in
DeepDive.

Figure~\ref{fig:e2e} shows the execution time of statistical inference
and learning for each update on different systems. We see from
Figure~\ref{fig:e2e} that \textsc{Incremental} achieves a $7\times$ to
$112\times$ speedup for News across all categories of rules. The
analysis rule A1 achieves the highest speedup -- this is not
surprising because, after applying A1, we do not need to rerun
statistical learning, and the updated distribution does not change
compared with the original distribution, so the sampling approach has
a 100\% acceptance rate. The execution of rules for feature extraction
(FE1, FE2), supervision (S1, S2), and inference (I1) has a $10\times$
speedup.  For these rules, the speedup over \textsc{Rerun} is to be
attributed to the fact that the materialized graph contains only
$10\%$ of the factors in the full original graph.  Below, we show that
both the sampling approach and variational approach contribute to the
speed-up.  Compared with A1, the speedup is smaller because these
rules produce a factor graph whose distribution changes more than
A1. Because the difference in distribution is larger, the benefit of
incremental evaluation is lower.

The execution of other KBC applications showed similar speedups, but
there are also several interesting data points. For Pharmacogenomics,
rule I1 speeds-up only $3\times$. This is caused by the fact that I1
introduces many new factors, and the new factor graph is 1.4$\times$
larger than the original one. In this case, \systemx needs to evaluate
those new factors, which is expensive. For Paleontology, we see that
the analysis rule A1 gets a $10\times$ speed-up because as illustrated
in the corpus statistics (Figure~\ref{tab:datasets}), the Paleontology
factor graph has fewer factors for each variable than other
systems. Therefore, executing inference on the whole factor graph is
cheaper.\\[-2pt]

\noindent
\textbf{Materialization Time} 
One factor that we need to consider is the materialization time for
\textsc{Incremental}.  \textsc{Incremental} took 12 hours to complete
the materialization (2000 samples), for each of the five systems.
Most of this time is spent in getting $2\times$ more samples than for
a single run of \textsc{Rerun}. We argue that paying this cost is
worthwhile given that it is a one-time cost and the materialization
can be used for many successive updates, amortizing the one-time cost.

\begin{figure}[t]
\vspace{-.75em} 
\centering
\includegraphics[width=0.4\textwidth]{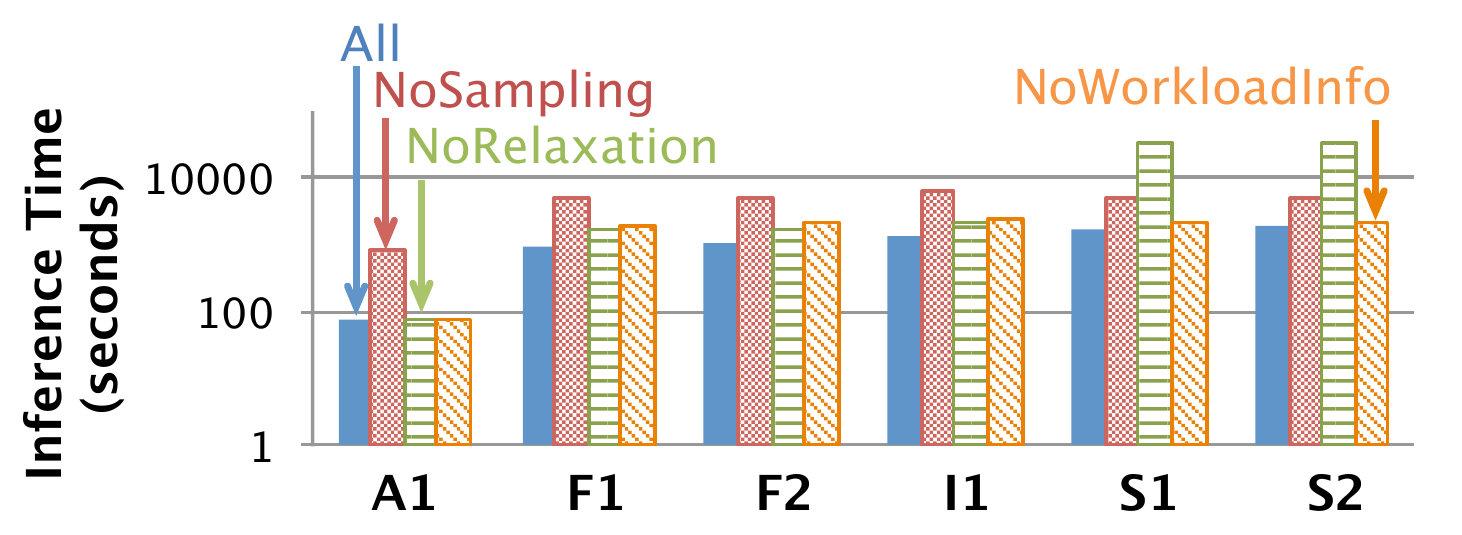}
\caption{Study of the tradeoff space on News.}
\label{fig:lesion}
\vspace{-.75em} 
\end{figure}


\subsection{Lesion Studies}

We conducted lesion studies to verify the effect of the tradeoff space
on the performance of \systemx. In each lesion study, we disable a
component of \systemx, and leave all other components untouched. We
report the execution time for statistical inference and learning.

We evaluate the impact of each materialization strategy on the final
end-to-end performance. We disabled either the sampling approach or
the variational approach and left all other components of the system
untouched. Figure~\ref{fig:lesion} shows the results for
News. Disabling either the sampling approach or the variational
approach slows down the execution compared to the ``full'' system. For
analysis rule A1, disabling the sampling approach leads to a more than
$11\times$ slow down, because the sampling approach has, for this
rule, a 100\% acceptance rate because the distribution does not
change. For feature extraction rules, disabling the sampling approach
slows down the system by $5\times$ because it forces the use of the
variational approach even when the distribution for a group of
variables does not change. For supervision rules, disabling the
variational approach is 36$\times$ slower because the introduction of
training examples decreases the acceptance rate of the sampling
approach.\\[-2pt]

\noindent
\textbf{Optimizer}
Using different materialization strategies for different groups of
variables positively affects the performance of \systemx. We compare
\textsc{Incremental} with a strong baseline \textsc{NoWorkloadInfo}
which, for each group, first runs the sampling approach. After all
samples have been used, we switch to the variational approach. Note
that this baseline is stronger than the strategy that fixes the same
strategy for all groups.  Figure~\ref{fig:lesion} shows the results of
the experiment. We see that with the ability to choose between the
sampling approach and variational approach according to the workload,
\systemx can be up to $2\times$ faster than \textsc{NoWorkloadInfo}.

\section{Conclusion} \label{sec:conclusion}
We described the \systemx approach to KBC and our experience building
KBC systems over the last few years. To improve quality, we argued
that a key challenge is to accelerate the development loop. We
described the semantic choices that we made in our language. By
building on SQL, DeepDive is able to use classical techniques to
provide incremental processing for the SQL components. However, these
classical techniques do not help with statistical inference, and we
described a novel tradeoff space for approximate inference
techniques. We used these approximate inference techniques to improve
end-to-end execution time in the face of changes both to the program
and the data; they improved system performance by two orders of
magnitude in five real KBC scenarios while keeping the quality high
enough to aid in the development process.

\par~\\ 
{\scriptsize
\noindent
{\bf Acknowledgments.}
We gratefully acknowledge the support of
the Defense Advanced Research Projects Agency (DARPA) XDATA program under No.\ FA8750-12-2-0335 and
DEFT program under No.\ FA8750-13-2-0039,
DARPA's MEMEX program and SIMPLEX program,
the National Science Foundation (NSF) CAREER Award under No.\ IIS-1353606,
the Office of Naval Research (ONR) under awards No.\ N000141210041 and No.\ N000141310129,
the National Institutes of Health Grant U54EB020405 awarded by the National Institute of Biomedical Imaging and Bioengineering (NIBIB) through funds provided by the trans-NIH Big Data to Knowledge (BD2K) initiative,
the Sloan Research Fellowship,
the Moore Foundation,
American Family Insurance,
Google, and
Toshiba.
Any opinions, findings, and conclusions or recommendations expressed in this material are those of the authors and do not necessarily reflect the views of DARPA, AFRL, NSF, ONR, NIH, or the U.S. government.
}
{\scriptsize \bibliographystyle{abbrv} \bibliography{inc}}

\balancecolumns
\clearpage
\appendix

\section{Additional Theoretical Details}

We find that the three semantics that we defined,
namely, Linear, Logical, and Ratio, have impact
to the convergence speed of the sampling procedure.
We describe these results in Section~\ref{app:rs}
and provide proof in Section~\ref{app:proof}.

\subsection{Convergence Results} \label{app:rs}

We describe our results on convergence rates, which are summarized in
Figure~\ref{tab:convergenceRates}. We first describe results for the
Voting program from Example~\ref{ex:voting}, and summarize those
results in a theorem. We now introduce the standard metric of convergence
in Markov Chain theory~\cite{Jerrum:1993:SIAM}.

\begin{definition}
The \emph{total variation distance} between two probability measures
$\mathbb{P}$ and $\mathbb{P}'$ over the same sample space $\Omega$ is
defined as
\[
\| \mathbb{P} - \mathbb{P}' \|_{\mathrm{tv}}
=
\sup_{A \subset \Omega} | \mathbb{P}(A) - \mathbb{P}'(A) |,
\]
that is, it represents the largest difference in the probabilities that
$\mathbb{P}$ and $\mathbb{P}'$ assign to the same event.
\end{definition}

In Gibbs sampling, the total variation distance is  a metric
of comparison between the actual distribution $\mathbb{P}_k$ achieved
at some timestep $k$ and the equilibrium distribution $\pi$.

\paragraph*{Types of Programs} 
We consider two families of programs motivated by our work in KBC:
voting programs and hierarchical programs.  The voting program is used
to continue the example, but some KBC programs are hierarchical (13/14
KBC systems from the literature). We consider some enhancements that
make our programs more realistic. We allow every tuple (IDB or EDB) to
have its own weight, i.e., we allow every ground atom $R(a)$ to have a
rule $w_a\; : R(a)$. Moreover, we allow any atom to be marked as
evidence (or not), which means its value is fixed. We assume that all
weights are {\em not} functions of the number of variables (and so are
$O(1)$). Finally, we consider two types of programs. The first we call
{\em voting programs}, which are a generalization of
Example~\ref{ex:voting} in which any subset of variables may be
evidence and the weights may be distinct.

\begin{proposition}
\label{propConvergenceRates}
Let $\varepsilon > 0$. Consider Gibbs sampling running on a particular
factor graph with $n$ variables, there exists a $\tau(n)$ that satisfies the
upper bounds in Figure~\ref{tab:convergenceRates} and such that $\| \mathbb{P}_k -
\pi \|_{\mathrm{tv}} \le \varepsilon$ after $\tau(n) \log(1 +
\varepsilon^{-1})$ steps.  Furthermore, for each of the classes, we
can construct a graph on which the minimum time to achieve total
variation distance $\varepsilon$ is at least the lower bound in
Figure~\ref{tab:convergenceRates}.
\end{proposition}

For example, for the voting example with logical semantics, $\tau(n) = O(n
\log n)$. The lower bounds are demonstrated with explicit examples and
analysis. The upper bounds use the special form of Gibbs sampling on
these factor graphs, and then use a standard argument based on 
\emph{coupling}~\cite{coupling} and an analysis very similar
to a generalized coupon collector process. This argument is sufficient to analyze the voting program in
all three semantics. 


We consider a generalization of the voting program in which each tuple
is a random variable (possibly with a weight); specifically, the program
consists of a variable $Q$, a set of ``Up'' variables
$U$, and a set of ``Down'' variables $D$.  Experimentally, we
compute how different semantics converges on the voting program as
illustrated in Figure~\ref{fig:semantic}.  In this figure, we vary the
size of $|U| + |D|$ with $|U| = |D|$ and all variables to be non-evidence
variables and measure the time that Gibbs sampling converges
to 1\% within the correct marginal probability of $Q$.  We see
that empirically, these semantics do have an effect on Gibbs sampling
performance (Linear converges much slower than either Ratio or Logical)
and we seek to give some insight into this phenomenon for
KBC programs.

We analyze more general \emph{hierarchical} programs in an asymptotic
setting.

\begin{figure}[t]
\scriptsize
\centering
\begin{tabular}{l|r|c|c}
\hline
Problem & OR Semantic & Upper Bound & Lower Bound \\
\hline
\multirow{3}{*}{Voting} & Logical & $O(n \log n)$ & $\Omega(n \log n)$ \\
 & Ratio      & $O(n \log n)$ & $\Omega(n \log n)$ \\
 & Linear   & $2^{O(n)}$ & $2^{\Omega(n)}$ \\
\hline
\end{tabular}
\caption{Bounds on $\tau(n)$. The $O$ notation hides constants that
  depend the query and the weights.}
\label{tab:convergenceRates}
\end{figure}

\begin{definition}
A rule $q(\bar x) \ts p_1(\bar x_1), \dots, p_k(x_k)$ is
hierarchical if either $\bar x = \emptyset$ or there is a single
variable $x \in \bar x$ such that $x \in \bar x_i$ for $i=1,\dots,
k$. A set of rules (or a Datalog program) is hierarchical if each rule
is hierarchical and they can be stratified.
\end{definition}

Evaluation of hierarchical programs is $\sharp\mathsf{P}$-hard, e.g.,
any Boolean query is hierarchical, but there are $\sharp\mathsf{P}$-hard Boolean
queries~\cite{Suciu:2011:Book}. Using a
similar argument, we show that any set of hierarchical rules that have
non-overlapping bodies converges in time $O(N \log N \log(1 +
\varepsilon^{-1}))$ for either Logical or Ratio semantics, where $N$ is the
number of factors. This statement is not trivial, as we show that Gibbs sampling
on the simplest non-hierarchical programs may take exponential time to converge.
Our definition is more general than the typical notions of {\em
safety}~\cite{Dalvi:2013:JACM}: we consider multiple rules with rich
correlations rather than a single query over a tuple independent database.
However, we only provide guarantees about sampling (not exact evaluation) and
have no dichotomy theorem.

We also analyze an asymptotic setting, 
in which the domain size grows, and we show that
hierarchical programs also converge in polynomial time in the Logical
semantics. This result is a theoretical example that suggests that the
more data we have, the better Gibbs sampling will behave. The key technical
idea is that for the hierarchical programs no variable has
unbounded influence on the final result, i.e., no variable contributes
an amount depending on $N$ to the final joint probability.

\begin{figure}[t]
\centering
\includegraphics[width=0.3\textwidth]{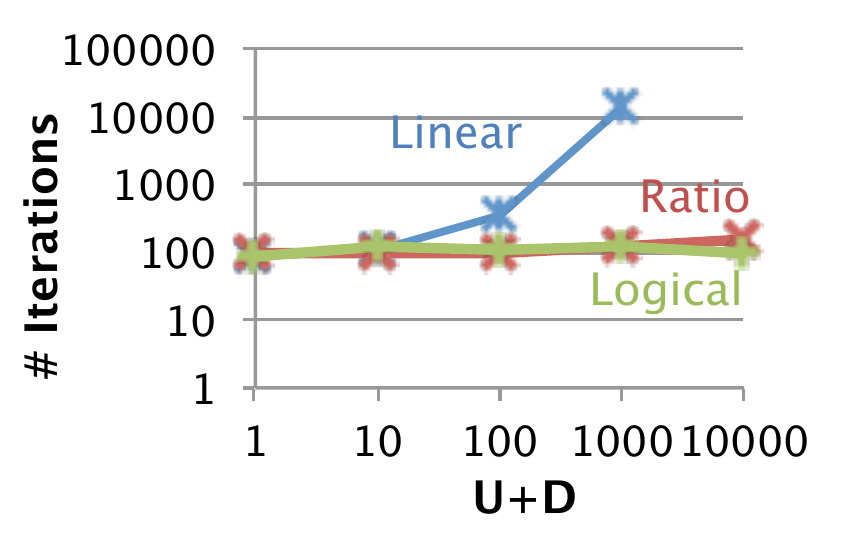}
\caption{Convergence of different semantics.}
\label{fig:semantic}
\end{figure}

\subsection{Proofs of Convergence Rates} \label{app:proof}

Here, we prove the convergence rates stated in Proposition
\ref{propConvergenceRates}.  The strategy for the upper bound proofs
involves constructing a \emph{coupling} between the Gibbs sampler
and another process that attains the equilibrium distribution at each
step.  First, we define a coupling~\cite{coupling,Levin:2006:Book}.
\begin{definition}[Coupling]
A coupling of two random variables $X$ and $X'$ defined on some separate
probability spaces $\mathbb{P}$ and $\mathbb{P}'$ is any new probability
space $\hat{\mathbb{P}}$ over which there are two random variables
$\hat X$ and $\hat X'$ such that $X$ has the same distribution as $\hat X$
and $X'$ has the same distribution as $\hat X'$.
\end{definition}
Given a coupling of two sequences of random variables $X_k$ and
$X_k'$, the \emph{coupling time} is defined as the first time $T$
when $\hat X_k = \hat X_k'$.  The following theorem lets us bound the
total variance distance in terms of the coupling time.
\begin{theorem}
\label{thmCouplingTime}
For any coupling $(\hat X_k, \hat X_k')$ with coupling time $T$,
\[
  \| \mathbb{P}(X_k \in \cdot) - \mathbb{P}'(X_k' \in \cdot) \|_{\mathrm{tv}}
  \le
  2 \hat{\mathbb{P}}(T > k).
\]
\end{theorem}
All of the coupling examples in this section use a \emph{correlated flip
coupler}, which consists of two Gibbs samplers $\hat X$ and $\hat X'$, each
of which is running with the same random inputs.  Specifically, both
samplers choose to sample the same variable at each timestep.  Then,
if we define $p$ as the probability of sampling $1$ for $\hat X$, and
$p'$ similarly, it assigns both variables to $1$ with probability
$\min(p, p')$, both variables to $0$ with probability $\min(1-p, 1-p')$,
and assigns different values with probability $|p - p'|$.  If we 
initialize $\hat X$ with an arbitrary distribution and $\hat X'$ with
the stationary distribution $\pi$, it is
trivial to check that $\hat X_k$ has distribution $\mathbb{P}_k$, and
$\hat{X}_k'$ always has distribution $\pi$.  Applying Theorem
\ref{thmCouplingTime} results in
\[
\| \mathbb{P}_k - \pi \|_{\mathrm{tv}}
\le
2 \hat{\mathbb{P}}(T > k).
\]
Now, we prove the bounds in Figure \ref{tab:convergenceRates}.

\begin{theorem}[UB for Voting: Logical and Ratio]
  For the voting example, if all the weights on the variables are bounded
  independently of the number of variables $n$, and $|U| = \Omega(n)$, and $|D| = \Omega(n)$,
  then for either the logical or
  ratio projection semantics, there exists a $\tau(n) = O(n \log n)$ such that for any
  $\varepsilon > 0$,
  $\| \mathbb{P}_k - \pi \|_{\mathrm{tv}} \le \varepsilon$
  for any $k \ge \tau(n) \log(\epsilon^{-1})$.
\end{theorem}
\begin{proof}
  For the voting problem, if we let $f$ denote the projection semantic we are using,
  then for any possible world $I$ the weight can be written as
  \[
    W
    =
    \sum_{u \in U \cap I} w_u + \sum_{d \in D \cap I} w_d
    + w \sigma f(|U \cap I|) - w \sigma f(|D \cap I|),
  \]
  where $\sigma = \mathsf{sign}(Q, I)$.  
  Let $I$ denote the world at the current timestep, $I'$ be the world at the next timestep,
  and $I_-$ be $I$ with the variable we choose to sample removed.
  Then if we sample a variable $u \in U$, and let $m_U = |U \cap I_-|$, then
  \begin{align*}
    \mathbb{P}\left(u \in I'\right)
    &=
    \frac{\exp\left(w_u + w \sigma f(m_U + 1)\right)}
      {\exp\left(w_u + w \sigma f(m_U + 1)\right) + \exp\left(w \sigma f(m_U)\right)} \\
    &=
    \left(1 + \exp\left(- w_u - w \sigma (f(m_U + 1) - f(m_U))\right)\right)^{-1}.
  \end{align*}
  Since $f(x+1) - f(x) < 1$ for any $x$, it follows that $\sigma (f(m_U+1) - f(m_U)) \ge -1$.
  Furthermore, since $w_u$ is bounded independently of $n$, we know
  that $w_u \ge w_{\mathrm{min}}$ for some constant $w_{\mathrm{min}}$.
  Substituting this results in
  \[
    \mathbb{P}\left(u \in I'\right)
    \ge
    \frac{1}{1 + \exp\left(- w_{\mathrm{min}} + w\right)}
    =
    p_{\mathrm{min}},
  \]
  for some constant $p_{\mathrm{min}}$ independent of $n$.  The same
  argument will show that, if we sample $d \in D$,
  \[
    \mathbb{P}\left(d \in I'\right)
    \ge
    p_{\mathrm{min}}.
  \]
  Now, if $|U \cap I| > \frac{p_{\mathrm{min}} |U|}{2}$, it follows that, for
  both the logical and ratio semantics,
  $f(m_U+1) - f(m_U) \le \frac{2}{p_{\mathrm{min}} |U|}$.
  Under these conditions, then if we sample $u$,
  we will have
  \begin{align*}
    & \frac{1}{1 + \exp\left(- w_u + w \frac{2}{p_{\mathrm{min}} |U|}\right)} \\
    &\hspace{2em}\le
    \mathbb{P}\left(u \in I'\right) \\
    &\hspace{2em}\le
    \frac{1}{1 + \exp\left(- w_u - w \frac{2}{p_{\mathrm{min}} |U|}\right)},
  \end{align*}
  The same argument applies to sampling $d$.  

  Next, consider the situation if we sample $Q$.
  \[
    \mathbb{P}\left(Q \in I'\right)
    =
    \left(
      1
      +
      \exp\left(-2w f(|U \cap I|) + 2w f(|D \cap I|)\right)
    \right)^{-1}.
  \]
  Under the condition that $|U \cap I| > \frac{p_{\mathrm{min}} |U|}{2}$, and similarly for 
  $D$, this is bounded by
  \begin{align*}
    &\left(1 + \exp\left(-2w f\left(\frac{p_{\mathrm{min}} |U|}{2}\right) + 2w f(|D|)\right)\right)^{-1} \\
    &\hspace{2em}\le
    \mathbb{P}\left(Q \in I'\right) \\
    &\hspace{2em}\le
    \left(1 + \exp\left(2w f\left(\frac{p_{\mathrm{min}} |D|}{2}\right) - 2w f(|U|)\right)\right)^{-1}.
  \end{align*}
  But, for both logical and ratio semantics, we have that
  $f(cn) - f(n) = O(1)$, and so for some $q_+$ and $q_-$ independent of $n$,
  \[
    q_- \le \mathbb{P}\left(Q \in I'\right) \le q_+.
  \]

  Now, consider a correlated-flip coupler running for $4 n \log n$ steps on the
  Gibbs sampler.  Let $E_1$ be the event that, after the first $2 n \log n$
  steps, each of the variables has been sampled at least once.  (This will have
  probability at least $\frac{1}{2}$ by the coupon collector's problem.)
  Let $E_U$ be the event that, after the first $2 n \log n$ steps,
  $|U \cap I| \ge \frac{p_\mathrm{min} |U|}{2}$ for the next $2 n \log n$ steps
  for both samplers,
  and similarly for $E_D$.  After all
  the entries have been sampled, $|U \cap I|$ and $|D \cap I|$ will each be
  bounded from below by a binomial random variable with parameter $p_\mathrm{min}$
  at each timestep, so from Hoeffding's inequality, the probability that
  this constraint will
  be violated by a sampler at any particular step is less than
  $\exp(-\frac{1}{2} p_{\mathrm{min}} |U|)$.  Therefore,
  \[
    \hat{\mathbb{P}}(E_U|E_1)
    \ge
    \left(1 - \exp\left(-\frac{1}{2} p_{\mathrm{min}} |U| \right) \right)^{4 n \log n}
    =
    \Omega(1).
  \]
  The same argument shows that $E_D = \Omega(1)$.
  Let $E_2$ be the event that all the variables are resampled at least
  once between time $2 n \log n$ and $4 n \log n$.  Again this event has 
  probability at least $\frac{1}{2}$.  Finally, let $C$ be the event that
  coupling occurs at time $4 n \log n$.  Given $E_1$, $E_2$, $E_U$, and $E_D$,
  this probability is equal to the probability that each variable
  coupled individually the last time it was sampled.  For $u$,
  this probability is
  \[
    1 - \frac{1}{1 + \exp\left(- w_u - w \frac{2}{p_{\mathrm{min}} n}\right)}
    + \frac{1}{1 + \exp\left(- w_u + w \frac{2}{p_{\mathrm{min}} n}\right)},
  \]
  and similarly for $d$.
  We know this probability is $1 - O(\frac{1}{n})$ by Taylor's theorem.  Meanwhile,
  for $Q$, the probability of coupling is $1 - q_+ + q_-$, which is always
  $\Omega(1)$.  Therefore,
  \[
    \hat{\mathbb{P}}(C|E_1,E_2,E_U,E_D)
    \ge
    \Omega(1) \left(1 - O\left(\frac{1}{n}\right) \right)^{n-1}
    =
    \Omega(1),
  \]
  and since
  \[
    \hat{\mathbb{P}}(C)
    =
    \hat{\mathbb{P}}(C|E_1,E_2,E_U,E_D) 
    \hat{\mathbb{P}}(E_U|E_1)
    \hat{\mathbb{P}}(E_D|E_1)
    \hat{\mathbb{P}}(E_2|E_1)
    \hat{\mathbb{P}}(E_1),
  \]
  we can conclude that $\hat{\mathbb{P}}(C) = \Omega(1)$.

  Therefore, this process couples with at least some constant probability $p_C$
  after $4 n \log n$ steps.  Since we can run this coupling argument
  independently an arbitrary number of times, it follows that, after
  $4 L n \log n$ steps, the probability of coupling will be at least
  $1 - (1 - p_C)^L$.  Therefore,
  \[
    \| \mathbb{P}_{4 L n \log n} - \pi \|_{\mathrm{tv}}
    \le
    \hat{\mathbb{P}}(T > 4 L n \log n) \le (1 - p_C)^L.
  \]
  This will be less than $\epsilon$ when
  $L \ge \frac{\log(\epsilon)}{\log(1 - p_C)}$, which occurs when
  $k \ge 4 n \log n \frac{\log(\epsilon)}{\log(1 - p_C)}$.
  Letting $\tau(n) = \frac{4 n \log n}{-\log(1 - p_C)}$ proves the
  theorem.
\end{proof}

\begin{theorem}[LB for Voting: Logical and Ratio]
  For the voting example, using either the logical or ratio projection
  semantics, the minimum time to achieve total variation distance $\varepsilon$
  is $O(n \log n)$.
\end{theorem}
\begin{proof}
  At a minimum, in order to converge, we must sample all the variables.  From
  the coupon collector's problem, this requires $O(n \log n)$ time.
\end{proof}

\begin{theorem}[UB for Voting: Linear]
  For the voting example, if all the weights on the variables are bounded
  independently of $n$, then for linear projection
  semantics, there exists a $\tau(n) = 2^{O(n)}$ such that for any
  $\varepsilon > 0$,
  $\| \mathbb{P}_k - \pi \|_{\mathrm{tv}} \le \varepsilon$
  for any $k \ge \tau(n) \log(\epsilon^{-1})$.
\end{theorem}
\begin{proof}
  If at some timestep we choose to sample variable $u$,
  \[
    \mathbb{P}\left(u \in I'\right)
    =
    \frac{\exp\left(w_u + w q\right)}{\exp\left(w_u + w q\right) + \exp(0)}
    =
    \Omega(1),
  \]
  where $q=\mathrm{sign}(Q,I)$. 
  The same is true for sampling $d$.  If we choose to sample $Q$,
  \begin{align*}
    &\mathbb{P}\left(Q \in I'\right) \\
    &\hspace{1em}=
    \frac{
      \exp\left(w |U \cap I| - w |D \cap I| \right)
    }{
      \exp\left(w |U \cap I| - w |D \cap I|\right)
      +
      \exp\left(-w |U \cap I| + w |D \cap I|\right)
    } \\
    &\hspace{1em}=
    \left(
      1
      +
      \exp\left(-2w (|U \cap I| - |D \cap I|) \right)
    \right)^{-1}
    =
    \exp(-O(n)).
  \end{align*}
  Therefore, if we run a correlated-flip coupler for $2 n \log n$ timesteps,
  the probability that coupling will have occured is greater the probability
  that all variables have been sampled (which is $\Omega(1)$) times
  the probability that all variables were sampled as $1$ in both samplers.
  Thus if $p_C$ is the probability of coupling,
  \[
    p_C 
    \ge
    \Omega(1) \left(\exp(-O(n))\right)^2 \left(\Omega(1)\right)^{2 O(n)}
    =
    \exp(-O(n)).
  \]
  Therefore, if we run for some $\tau(n) = 2^{O(n)}$ timesteps, coupling
  will have occurred at least once with high probability.  The statements
  in terms of total variance distance follow via the same argument used in
  the previous lower bound proof. 
\end{proof}

\begin{theorem}[LB for Voting: Linear]
  For the voting example using linear projection
  semantics, the minimum time to achieve total variation distance $\varepsilon$
  is $2^{O(n)}$.
\end{theorem}
\begin{proof}
  Consider a sampler with unary weights all zero
  that starts in the state $I = {Q} \cup {U}$. 
  We will show that it takes an exponential amount of time until $Q \notin I$.
  From above, the probability of flipping $Q$ will be
  \[
    \mathbb{P}\left(Q \notin I'\right)
    =
    \left(
      1
      +
      \exp\left(2w (|U \cap I| - |D \cap I|) \right)
    \right)^{-1}.
  \]
  Meanwhile, the probability to flip any $u$ while $Q \in I$ is
  \[
    \mathbb{P}\left(u \notin I'\right)
    =
    \frac{\exp\left(0\right)}{\exp\left(w\right) + \exp(0)}
    =
    \left(\exp\left(w\right) + 1\right)^{-1}
    =
    p,
  \]
  for some $p$, and the probability of flipping $d$ is similarly
  \[
    \mathbb{P}\left(d \in I'\right)
    =
    \frac{\exp\left(-w\right)}{\exp\left(-w\right) + \exp(0)}
    =
    p.
  \]
  Now, consider the following events which could happen at any timestep.
  While $Q \in I$,
  let $E_U$ be the event that $|U \cap I| \le (1 - 2p) n$, and let $E_D$ be
  the event that $|D \cap I| \ge 2pn$.  Since $|U \cap I|$ and $|D \cap I|$
  are both bounded
  by binomial random variables with parameters $1 - p$ and $p$ respectively,
  Hoeffding's inequality states that, at any timestep,
  \[
    \mathbb{P}\left(|U \cap I| \le (1 - 2p) n\right) \le \exp(-2 p^2 n),
  \]
  and similarly
  \[
    \mathbb{P}\left(|D \cap I| \ge 2p n\right) \le \exp(-2 p^2 n).
  \]
  Now, while these bounds are satisfied, let $E_Q$ be the event that $Q \notin I'$.
  This will be bounded by
  \[
    \mathbb{P}\left(E_Q\right)
    =
    \left(
      1
      +
      \exp\left(2w (1 - 4p) n \right)
    \right)^{-1}
    \le
    \exp(-2w (1 - 4p) n).
  \]
  It follows that at any timestep,
  $\mathbb{P}\left(E_U \vee E_D \vee E_Q\right) = \exp(-\Omega(n))$.
  So, at any timestep $k$, 
  $\mathbb{P}\left(Q \in I\right) = k \exp(-\Omega(n))$.  However, by symmetry
  the stationary distribution $\pi$ must have
  $\pi\left(Q \in I\right) = \frac{1}{2}$.  Therefore, the total
  variation distance is bounded by
  \[
    \| \mathbb{P}_k - \pi \|_{\mathrm{tv}}
    \ge
    \frac{1}{2} - k \exp(-\Omega(n)).
  \]
  So, for convergence to less than $\varepsilon$, we must require at least
  $2^{O(n)}$ steps.
\end{proof}

\section{Additional System Details} 

\subsection{Decomposition with Inactive Variables}
\label{sec:tradeoff:multi}

Our system also uses the following heuristic.  The intuition behind
this heuristic is that, the fewer variables we materialize, the faster
all approaches are. DeepDive allows the user to specify a 
set of rules (the ``interest area'') that identify
a set of relations that she will focus on in the next iteration of the
development. Given the set of rules, we use the standard {\em
  dependency graph} to find relations that could be changed by the
rules.  We call variables in the relations that the developer wants to
improve {\em active variables}, and we call {\em inactive variables}
those that are not needed in the next KBC iteration. Our previous
discussion assumes that all the variables are active. The intuition is
that we only require the active variables, but we first need to
marginalize out the inactive variables to create a new factor
graph. If done naively, this new factor graph can be excessively
large. The observation that serves as the basis of our optimization is
as follows: conditioning on all active variables, we can partition all
inactive variables into sets that are conditionally independent from
each other. Then, by appropriately grouping together some of the sets,
we can create a more succinct factor graph to materialize.

\paragraph*{Procedure}

We now describe in detail our procedure to decompose a graph in the
presence of inactive variables, as illustrated in
Algorithm~\ref{alg:decomp}. Let $\{\mathcal{V}^{(i)}_j\}$ be the
collection of all the sets of inactive variables that are
conditionally independent for all other inactive variables,
\emph{given the active variables}. This is effectively a partitioning
of all the inactive variables (Line 1). For each set
$\mathcal{V}^{(i)}_j$, let $\mathcal{V}^{(a)}_j$ be the \emph{minimal}
set of active variables conditioning on which the variables in
$\mathcal{V}^{(i)}_j$ are independent of all other inactive variables
(Line 2).  Consider now the set $\mathcal{G}$ of pairs
$(\mathcal{V}^{(i)}_j,\mathcal{V}^{(a)}_j)$. It is easy to see that we
can materialize each pair separately from other pairs. Following this
approach to the letter, some active random variables may be
materialized multiple times, once for each pair they belong to, with
an impact on performance. This can be avoided by grouping some pairs
together, with the goal of minimizing the runtime of inference on the
materialized graph. Finding the optimal grouping is actually
\textsf{NP-hard}, even if we make the simplification that each group
$g=(\mathcal{V}^{(i)}, \mathcal{V}^{(a)})$ has a cost function $c(g)$
that specifies the inference time, and the total inference time is the
sum of the cost of all groups. The key observation concerning the
NP-hardness is that two arbitrary groups $g_1=(\mathcal{V}^{(i)}_1,
\mathcal{V}^{(a)}_1)$ and $g_2=(\mathcal{V}^{(i)}_2,
\mathcal{V}^{(a)}_2)$ can be materialized together to form a new group
$g_3=(\mathcal{V}^{(i)}_1 \cup \mathcal{V}^{(i)}_2,
\mathcal{V}^{(a)}_1 \cup \mathcal{V}^{(a)}_2)$, and $c(g_3)$ could be
smaller than $c(g_1)+c(g_2)$ when $g_1$ and $g_2$ share most of the
active variables.  This allows us to reduce the problem to
\textsc{WeightedSetCover}. Therefore, we use a greedy heuristic that
starts with groups, each containing one inactive variable, and
iteratively merges two groups
$(\mathcal{V}^{(i)}_1,\mathcal{V}^{(a)}_1)$ and
$(\mathcal{V}^{(i)}_2,\mathcal{V}^{(a)}_2)$ if $|\mathcal{V}^{(a)}_1
\cup
\mathcal{V}^{(a)}_2|=\max\{|\mathcal{V}^{(a)}_1|,|\mathcal{V}^{(a)}_2|\}$
(Line 4-6). The intuition behind this heuristic is that, according to
our study of the tradeoff between different materialization
approaches, the fewer variables we materialize, the greater the speedup
for all materialization approaches will be.

\begin{algorithm}[t]
\scriptsize
\begin{algorithmic}[1]
\Require Factor graph $FG = (\mathcal{V},\mathcal{F})$,
active variables $\mathcal{V}^{(a)}$,
and inactive variables $\mathcal{V}^{(i)}$.
\Ensure Set of groups of variables $\mathcal{G} = \{(\mathcal{V}^{(i)}_i, \mathcal{V}^{(a)}_i)\}$ that will be materialized independently with other groups.
\algrule

\State $\{\mathcal{V}^{(i)}_j\} \leftarrow$ 
connected components of the factor graph
after removing all active variables, i.e., 
$\mathcal{V} - \mathcal{V}^{(a)}$.

\State $\mathcal{V}^{(a)}_j \leftarrow$ minimal set of active
variables conditioned on which $\mathcal{V}^{(i)}_j$ is independent
of other inactive variables.

\State $\mathcal{G} \leftarrow \{(\mathcal{V}^{(i)}_j, \mathcal{V}^{(a)}_j)\}$.

\ForAll{$j \ne k$ s.t. $|\mathcal{V}^{(a)}_j \cup
\mathcal{V}^{(a)}_k|=\max\{|\mathcal{V}^{(a)}_j|,|\mathcal{V}^{(a)}_k|\}$}
\State $\mathcal{G} \leftarrow
\mathcal{G} \cup \{(\mathcal{V}^{(i)}_j \cup \mathcal{V}^{(i)}_k, \mathcal{V}^{(a)}_j \cup \mathcal{V}^{(a)}_k)\}$
\State $\mathcal{G} \leftarrow 
\mathcal{G} - \{(\mathcal{V}^{(i)}_j, \mathcal{V}^{(a)}_j), (\mathcal{V}^{(i)}_k, \mathcal{V}^{(a)}_k)\}$
\EndFor

\end{algorithmic}
\caption{Heuristic for Decomposition with Inactive Variables}
\label{alg:decomp}
\end{algorithm}

\begin{figure}[t]
\centering
\includegraphics[width=0.3\textwidth]{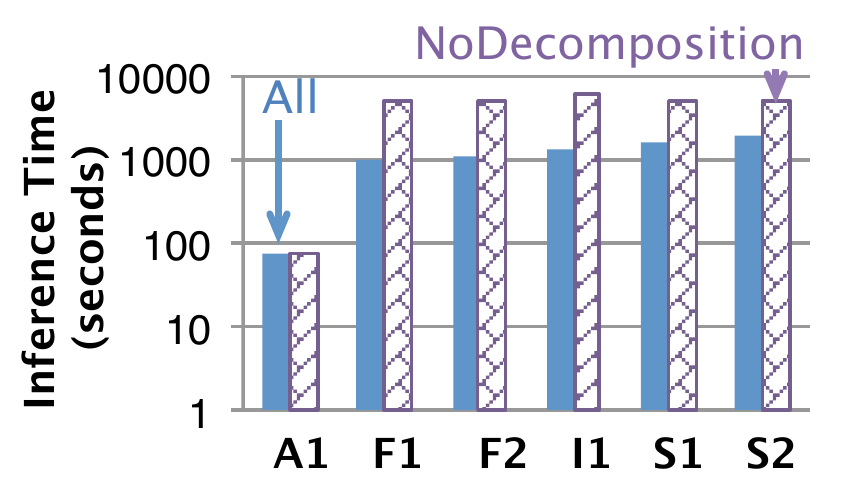}
\caption{The Lesion Study of Decomposition.}
\label{fig:lesion_decomp}
\end{figure}

\paragraph*{Experimental Validation}
We also verified that using the decomposition technique has a positive
impact on the performance of the system. We built
\textsc{NoDecomposition} that does not decompose the factor graph and
select either the sampling approach or the variational approach for
the whole factor graph. We report the {\em best} of these two
strategies as the performance of \textsc{NoDecomposition} in
Figure~\ref{fig:lesion_decomp}.  We found that, for A1, removing
structural decomposition does not slow down compared with \systemx,
and it is actually 2\% faster because it involves less computation for
determining sample acceptance. However, for both feature extraction
and supervision rules, \textsc{NoDecomposition} is at least $6\times$
slower. For both types of rules, the \textsc{NoDecomposition} costs
corresponds to the cost of the variational approach. In contrast, the
sampling approach has near-zero acceptance rate because of a large
change in the distribution.

\begin{figure}[t]
\centering
\begin{tabular}{rrrrrr}
\hline
& {\bf Adv.} & {\bf News} & {\bf Gen.} & {\bf Pharma.} & {\bf Paleo.} \\
\hline
{\bf \# Samples} & 2083 & 3007 & 22162 & 1852 & 2375 \\
\hline
\end{tabular}
\caption{Number of Samples We can Materialize in 8 Hours.}
\label{fig:app:mat}
\end{figure}

\subsection{Materialization Time} \label{sec:app:mat}

We show how materialization time changes across different systems. In
the following experiment, we set the total materialization time to 8
hours, a common time budget that our users often used for overnight
runs. We execute the {\em whole} materialization phase of \systemx for
all systems, which will materialize both the sampling approach and the
variational approach, and we require the whole process to finish in 8
hours.  Figure~\ref{fig:app:mat} shows how many samples we can collect
given this 8-hour budget: for all five systems, in 8 hours, we can
store more than 2000 samples, from which to generate 1000 samples
during inference if the acceptance rate for the sampling approach is
0.5.

\subsection{Incremental Learning} \label{sec:inclearn}

\begin{figure}[t]
\centering
\includegraphics[width=0.3\textwidth]{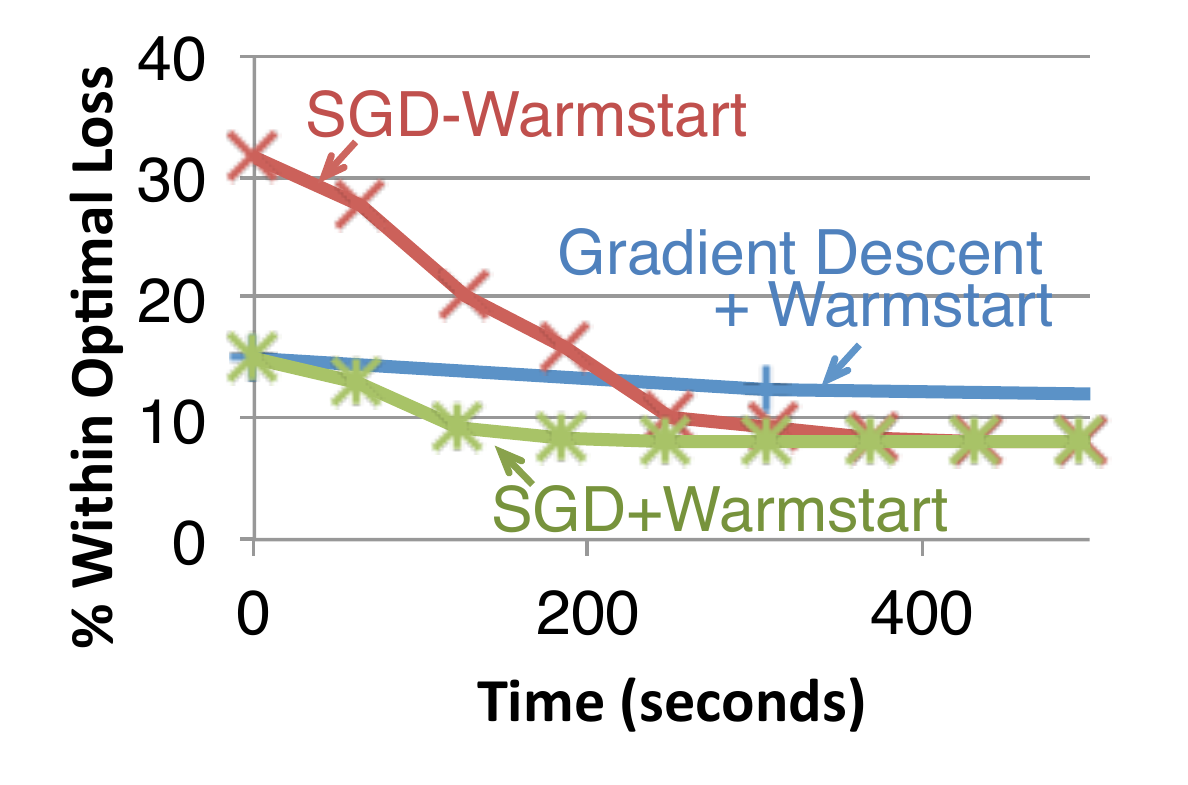}
\caption{Convergence of Different Incremental Learning Strategies.}
\label{fig:app:learn}
\end{figure}

Rerunning learning happens when one labels new documents (typically an
analyst in the loop), so one would like this to be an efficient
procedure. One interesting observation is that a popular technique
called distant supervision (used in DeepDive) is able to label new
documents heuristically using rules, so that new labeled data is
created without human intervention. Thus, we study how to
incrementally perform learning. This is well studied in the online
learning community, and here we adapt essentially standard
techniques. We essentially get this for free because we choose to
adapt standard online methods, such as stochastic gradient descent
with warmstart, for our training system. Here, warmstart means that
\systemx uses the learned model in the last run as the starting point
instead of initializing the model to start randomly.

\paragraph*{Experimental Validation}

We validate that, by adapting standard online learning methods,
\systemx is able to outperform baseline approaches. We compare
\systemx with two baseline approaches, namely (1) stochastic gradient
descent without warmstart and (2) gradient descent with warmstart. We
run on the dataset News with rules F2 and S2, introduces {\em both}
new features and new training examples with labels. We obtain a proxy
for the optimal loss by running both stochastic gradient descent and
gradient descent separately for 24 hours and picking the lowest
loss. For each learning approach, we grid search its step size in
\{1.0, 0.1, 0.01, 0.001, 0.0001\}, run for 1 hour, and pick the
fastest one to reach a loss that is within 10\% of the optimal
loss. We estimate the loss after each learning epoch and report the
percentage within the optimal loss in Figure~\ref{fig:app:learn}.

As shown in Figure~\ref{fig:app:learn}, the loss of all learning
approaches decrease, with more learning epochs. However, \systemx's
SGD+Warmstart approach reaches a loss that achieves a loss within 10\%
optimal loss faster than other approaches.  Compared with
SGD-Warmstart, SGD+Warmstart is 2$\times$ faster in reaching a loss
that is within 10\% of the optimal loss because it starts with a model
that has lower loss because of warmstart. Compared with gradient
descent, SGD+Warmstart is about 10$\times$ faster. Although
SGD+Warmstart and gradient descent with warmstart have the same initial
loss, SGD converges faster than gradient descent in our example.

\subsection{Discussion of Concept Drift}

\begin{figure}[t]
\centering
\includegraphics[width=0.3\textwidth]{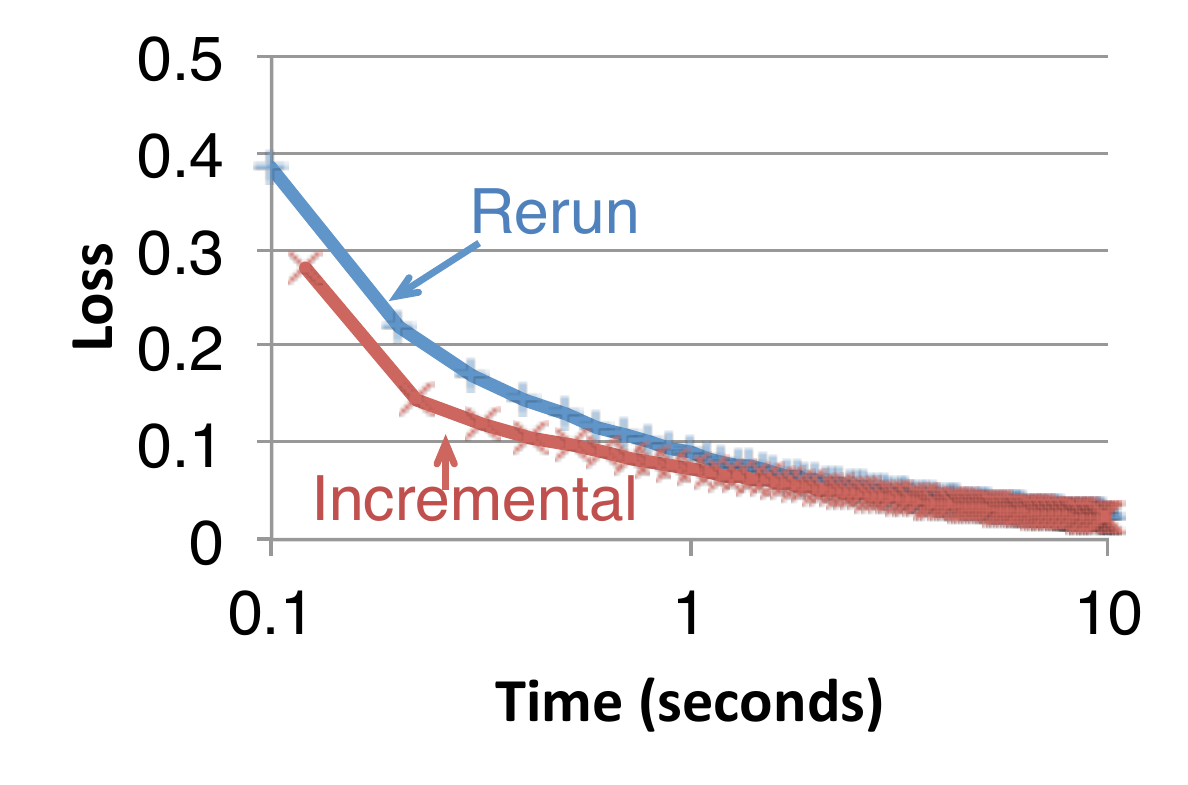}
\caption{Impact of Concept Drifts.}
\label{fig:app:learn_concept_drift}
\end{figure}

Concept drift~\citeapp{Gama:2014:ACMCS,Katakis:2009:JIIS}
  refers to the scenario that arises in online learning when the
  training examples become available as a stream, and the distribution
  from which these examples are generated keeps changing
  overtime. Therefore, the key technique in how to resolve the problem
  of concept drift is premarily in how to design a machine learning
  system to automatically {\em detect} and {\em adapt to} the
  changes. The design needs to consider how to not only update the
  model incrementally, but also ``forget the old
  information''~\citeapp{Gama:2014:ACMCS}.  Although we have not
  observed quality issues caused by concept drift in our KBC
  applications, it is possible for DeepDive to encounter this issue in
  other applications. Thus, we study the impact of concept drift to
  our incremental inference and learning approach. 

\paragraph*{Impact on Performance}

In this work, we solely focus on the efficiency of updating the model
incrementally. When concept drift happens, the new distribution and
new target model would be significantly different from the
materialized distribution as well as the learned model from last
iteration.  From our analysis in Section~\ref{sec:tradeoff} and
Section~\ref{sec:inclearn}, we hypothesize there are two impacts of
concept drift that require consideration in DeepDive. First, the
difference between materialized distribution and the target
distributed is modeled as the amount of change in
Section~\ref{sec:tradeoff}, and therefore, in the case of concept
drift, we expect our system favors variational approach more than
sampling approach. Second, because the difference between the learned
model and target model is large, it is not clear whether the warmstart
scheme we discussed in Section~\ref{sec:inclearn} will work.  

\paragraph*{Experimental Validation}
As our workloads do not have concept drift, we follow prior
art~\citeapp{Katakis:2009:JIIS}, we use its dataset that contains
9,324 emails ordered chronologically, and predict exactly the same
task that for each email predict whether it is spam or not.  We
implement a logistic regression classifier in DeepDive with a rule
similar to Example~\ref{exp:lr}. We use the first 10\% emails and
first 30\% emails as training sets, respectively, and use the
remaining 70\% as the testing set. We construct two systems in a
similar way as Section~\ref{sec:exp}: \textsc{Rerun} that trains
directly on 30\% emails, and \textsc{Incremental} that materialize
using 10\% emails and incrementally train on 30\% emails. We run both
systems and measure the test set loss after each learning epoch.  

Figure~\ref{fig:app:learn_concept_drift} shows the result. We see that
even with concept drift, both systems converge to the
same loss. \textsc{Incremental} converges faster than
\textsc{Rerun}, and it is because even with
concept drift, warmstart still allows \textsc{Incremental}
to start from a lower loss at the first iteration.
In terms of the time used for each iteration, 
\textsc{Incremental} and \textsc{Rerun} uses roughly
the same amount of time for each iteration. This is 
expected because \textsc{Incremental} rejects almost
all samples due to the large difference between distributions
and switch to variational approach, which, for a logistic
regression model where all variables are active and all
weights have been changed, is very similar to the original model.
From these result we see that, even with concept drifts,
\textsc{Incremental} still provides benefit over
rerunning the system from scratch due to warmstart; however, 
as expected, the benefit from incremental
inference is smaller.

\section{Additional Related Work}
We provide additional related work, especially those work that impacts
the current design of DeepDive's language component.

\paragraph*{More Related KBC Systems}

DeepDive's model of KBC is motivated by the recent attempts
of using machine learning-based technique for KBC~\citeapp{Yates:2007:NAACL,Banko:2007:IJCAI,Barbosa:ICDE:2013,Kasneci:2009:SIGMODR,Suchanek:2009:WWW,Poon:2007:AAAI,Zhu:2009:WWW}
and the line of research that aims to improve the quality
of a specific component of KBC system~\citeapp{
mintz2009distant,riedel2010modeling,yao2010collective,min2013distant,zhangtowards,zhang2012big,fan2014distant,krause2012large,nguyen2011end,surdeanu2011stanford,purver2012experimenting,marchetti2012learning,surdeanu2010simple,hoffmann2010learning,bunescu2007learning,craven1999constructing,
li2014weakly}. When designing DeepDive, we used these
systems as test cases to justify
the generality of our framework. In fact, we find that
DeepDive is able to model 13 of these popular KBC 
systems~\citeapp{mintz2009distant,yao2010collective,riedel2010modeling,zhangtowards,fan2014distant,min2013distant,krause2012large,nguyen2011end,purver2012experimenting,marchetti2012learning,zhang2012big,li2014weakly,blunsom2014convolutional}.

\paragraph*{Other Incremental Algorithms} We build on decades of study of
incremental view
maintenance~\citeapp{Gupta:1999:Book,Chirkova:2012:Book}: we rely on the
classic incremental maintenance techniques to handle relational
operations. Recently, others have proposed different approaches for
incremental maintenance for individual analytics workflows like
iterative linear algebra problems~\citeapp{Nikolic:2014:SIGMOD},
classification as new training data arrived~\citeapp{Koc:2011:PVLDB}, or
segmentation as new data arrives~\citeapp{Chen:2008:ICDE}.

Factor graphs have been used by some probabilistic databases as the
underlying representation~\citeapp{Sen:2009:VLDBJ,Wick:2010:PVLDB}, but
they did not study how to reuse computation across different, but
similar, factor graphs.

{\scriptsize
\bibliographystyleapp{abbrv}
\bibliographyapp{inc}}

\end{document}